\documentclass[a4paper,UKenglish,cleveref, autoref, thm-restate]{lipics-v2021}

\pdfoutput=1 

\title{Arboricity Matters in Triangle Counting with Random Edges} 

\author{Arijit Bishnu}{Indian Statistical Institute, Kolkata, India \and \url{https://www.isical.ac.in/~arijit/}}{arijit@isical.ac.in}{https://orcid.org/0000-0003-0018-7314}{}

\author{Debarshi Chanda}{Indian Statistical Institute, Kolkata, India \and \url{https://sites.google.com/view/debarshi-ch/}}{debarshi.chanda.1997@gmail.com}{https://orcid.org/0009-0001-2197-6366}{}

\author{Gopinath Mishra}{Institute of Mathematical Sciences, Chennai, India \and \url{https://sites.google.com/view/gopinathmishra}}{gopinath@imsc.res.in}{https://orcid.org/0000-0003-0540-0292}{}

\authorrunning{A. Bishnu, D. Chanda, G. Mishra}

\Copyright{Arijit Bishnu and Debarshi Chanda and Gopinath Mishra} 



\category{} 

\relatedversion{https://arxiv.org/abs/2502.15379} 

\nolinenumbers 

\EventEditors{Mohit Singh and Tom Gur}
\EventNoEds{2}
\EventLongTitle{Approximation, Randomization, and Combinatorial Optimization. Algorithms and Techniques (APPROX/RANDOM 2026)}
\EventShortTitle{\mbox{\scriptsize APPROX/RANDOM 2026}}
\EventAcronym{APPROX/RANDOM}
\EventYear{2026}
\EventDate{August 19--21, 2026}
\EventLocation{Boston University, Boston, Massachusetts, USA}
\EventLogo{}
\SeriesVolume{392}
\ArticleNo{55}

\usepackage{float}

\usepackage{amsmath, mathtools, amsthm}
\usepackage{algorithm}
\usepackage[noend]{algpseudocode}
\usepackage{amssymb}
\usepackage{thmtools}

\usepackage{hyperref}
\usepackage[most]{tcolorbox}
\newtcolorbox{idea}[1][]
{
colbacktitle=cyan,
colback=cyan!10,
arc=1pt,
boxrule=1pt,
title=#1 
}

\newtcolorbox{update}[1][]
{
colbacktitle=gray,
colback=gray!10,
arc=1pt,
boxrule=1pt,
title=#1 
}

\newtcolorbox{question}[1][]
{
coltitle=black,
colbacktitle=yellow,
colback=yellow!10,
arc=1pt,
boxrule=1pt,
title=#1 
}

\newtcolorbox{answer}[1][]
{
coltitle=black,
colbacktitle=violet!10,
colback=violet!5,
arc=1pt,
boxrule=1pt,
title=#1 
}

\newcommand{\sA}{\mathcal{A}}
\newcommand{\sB}{\mathcal{B}}

\newcommand{\sD}{\mathcal{D}}
\newcommand{\sE}{\mathcal{E}}

\newcommand{\sW}{\mathcal{W}}

\newcommand{\arboricity}{\alpha}
\newcommand{\degeneracy}{\kappa}
\newcommand{\arboricitygraph}[1]{\alpha_{#1}}
\newcommand{\numtriangle}{T}
\newcommand{\esttriangle}{\widetilde{T}}
\newcommand{\emptriangle}{\widehat{T}}
\newcommand{\threshold}{\tau}
\newcommand{\upperthreshold}{h}
\newcommand{\lowerthreshold}{l}
\newcommand{\lighttriangles}[1]{\numtriangle^{#1}_{L}}
\newcommand{\heavytriangles}[1]{\numtriangle^{#1}_{H}}
\newcommand{\triangleset}{\Delta}
\renewcommand{\triangle}{\sf t} 
\newcommand{\weightfunc}{w}
\newcommand{\empweightfunc}{\widehat{\weightfunc}}
\newcommand{\weightfamily}{\sW}
\newcommand{\exactheavyoracle}{\texttt{HeavyOracle}}
\newcommand{\heavyoracle}{\texttt{Heavy}}
\newcommand{\samplededges}{S}
\newcommand{\sampledtriangles}{S_\triangleset}
\newcommand{\edgesamplesize}{s}
\newcommand{\querycount}{r}
\newcommand{\searchesttriangle}{\bar{t}}

\newcommand{\edgesensitivity}{\Delta_\edgeset}
\newcommand{\vertexsensitivity}{\Delta_\vertexset}

\newcommand{\bucket}{\sB}

\newcommand{\stringlength}{M}

\newcommand{\alicestring}{{x}}

\newcommand{\vertexsetA}{A}
\newcommand{\vertexsetAhat}{A'}
\newcommand{\vertexsetB}{B}
\newcommand{\vertexsetBhat}{B'}
\newcommand{\vertexA}{a}
\newcommand{\vertexAhat}{a'}
\newcommand{\vertexB}{b}
\newcommand{\vertexBhat}{b'}
\newcommand{\trianglemakerset}{R}
\newcommand{\trianglemaker}{r}

\newcommand{\rqcomplexity}[2]{R_{#1}\fbrac{#2}}
\newcommand{\rqgraph}{{\graph_{\alicestring}}}
\newcommand{\ptp}{\texttt{PTP}}
\newcommand{\ptpprob}{k}
\newcommand{\ptpsep}{\gamma}
\newcommand{\ptpdone}{\distribution_0}
\newcommand{\ptpdtwo}{\distribution_1}
\newcommand{\defptp}{\ptp_{\stringlength,\ptpprob,\ptpsep}}

\newcommand{\trianglealgo}{\algo_{\numtriangle}}
\newcommand{\trianglealgoptp}{\trianglealgo^{\ptp}}

\newcommand{\lbedgecount}{m}
\newcommand{\lbarboricity}{\arboricity}

\newcommand{\constant}{c}

\newcommand{\graph}{G}
\newcommand{\vertexset}{V}
\newcommand{\edgeset}{E}
\newcommand{\vertexcount}{n}
\newcommand{\edgecount}{m}
\newcommand{\vertex}{v}
\newcommand{\altvertex}{u}
\newcommand{\edge}{e}
\newcommand{\neighbour}{{\sf N}}
\newcommand{\degree}[1]{\mbox{deg}_{#1}}

\newcommand{\degreeq}{\texttt{Degree}}
\newcommand{\randedgeq}{\texttt{RandomEdge}}
\newcommand{\neighbourq}{\texttt{Neighbour}}
\newcommand{\edgeexistsq}{\texttt{Pair}}

\newcommand{\remove}[1]{}
\newcommand{\complain}[1]{\textcolor{violet}{#1}}

\newcommand{\field}[1]{\mathbb{#1}}

\newcommand{\Nat}{\field{N}}

\newcommand{\norm}[1]{\left\|{#1}\right\|}

\newcommand{\func}[3]{{#1} : {#2} \rightarrow {#3}}
\newcommand{\ceil}[1]{\left\lceil{#1}\right\rceil}

\newcommand{\size}[1]{\left|{#1}\right|}

\newcommand{\blue}[1]{\textcolor{blue}{#1}}

\newcommand{\fbrac}[1]{\left({#1}\right)}
\newcommand{\sbrac}[1]{\left\{{#1}\right\}}
\newcommand{\tbrac}[1]{\left[{#1}\right]}

\DeclareMathOperator*{\Exp}{\field{E}}
\DeclareMathOperator*{\Var}{\mathrm{Var}}

\newcommand{\distribution}{\sD}
\newcommand{\uniform}{\mathrm{Unif}}

\newcommand{\approxerror}{\varepsilon}
\newcommand{\confidence}{\delta}
\newcommand{\appcon}{\fbrac{\approxerror,\confidence}}

\newcommand{\bigo}[1]{O\fbrac{{#1}}}
\newcommand{\bigot}[1]{\widetilde{O}\fbrac{{#1}}}
\newcommand{\bigoted}[1]{\widetilde{O}_{\approxerror,\confidence}\fbrac{{#1}}}
\newcommand{\smallo}[1]{o\fbrac{{#1}}}

\newcommand{\bigomega}[1]{\Omega\fbrac{{#1}}}
\newcommand{\bigomegat}[1]{\widetilde{\Omega}\fbrac{{#1}}}

\newcommand{\algo}{\sA}
\newcommand{\poly}[1]{poly\fbrac{#1}}

\usepackage{soul}

\begin{document}

\maketitle


\begin{abstract}
Given a simple, unweighted, undirected graph $G=(V,E)$ with $|V|=n$ and $|E|=m$, and parameters $0 < \varepsilon, \delta <1$, along with \texttt{Degree}, \texttt{Neighbour}, \texttt{Pair} and \texttt{RandomEdge} query access to $G$, we provide a query-based randomized algorithm to generate an estimate $\widehat{T}$ of the number of triangles $T$ in $G$, such that $\widehat{T} \in [(1-\varepsilon)T , (1+\varepsilon)T]$ with probability at least $1-\delta$. The query complexity of our algorithm is $\widetilde{O}\left({m \alpha \log(1/\delta)}/{\varepsilon^3 T}\right)$, where $\alpha$ is the arboricity of $G$. Our work can be seen as a natural progression to the line of recent works [Eden et al., SIAM J Comp., 2017; Assadi et al., ITCS 2019; Eden et al., SODA 2020] that considered subgraph or triangle counting with or without the use of \texttt{RandomEdge} query. Of these works, Eden et al. [SODA 2020] considers the role of arboricity. Our work is the first to consider how \texttt{RandomEdge} query can leverage the structural property of arboricity. Furthermore, continuing in the line of work of Assadi et al. [APPROX/RANDOM 2022], we also provide a lower bound of $\widetilde{\Omega}\left({m \alpha \log(1/\delta)}/{\varepsilon^2 T}\right)$ that matches the upper bound exactly on arboricity, $\delta$, and almost on $\varepsilon$.\end{abstract}

\section{Introduction}

\label{sec:intro}
\noindent
Counting the number of triangles present in graphs is a classical problem \cite{DBLP:journals/siamcomp/ChibaN85,DBLP:journals/siamdm/AlonKKR08,SWExperimentalListingTriangles2005,AlonTestingFOCS2013,DBLP:conf/soda/Bar-YossefKS02}. Researchers in the streaming (sub-linear space) community have studied this problem for the last two decades; see~\cite{DBLP:conf/soda/Bar-YossefKS02,DBLP:conf/stacs/BeraC17,jayaram_et_al:LIPIcs.APPROX/RANDOM.2021.11, AhnGMPODS2012, TriangleCountingDynamicGraphALgorithmica2016, DBLP:journals/algorithmica/AlonYZ97, JhaSeshadriStreamingTransitivity, BuriolFLPODS06, CORMODE2017SecondLook, ShinDynamicStreamTriangle, KaneMelhornArbitrarySubgrapgICALP12} for a sample of such works. The problem has also been studied by researchers in the property testing (sub-linear time) community in the last decade~\cite{Dana_Ron_Triangle_Counting,DBLP:conf/pods/McGregorVV16,KallaugherP17,DBLP:conf/stacs/BeraC17,BeraSeshadriStreamingDegeneracy,DBLP:conf/soda/EdenRS20,jayaram_et_al:LIPIcs.APPROX/RANDOM.2021.11}. Here, graphs can be accessed only through certain queries. On the other hand, starting with the work of~\cite{DBLP:journals/siamcomp/ChibaN85}, the graph parameter called arboricity has also found its use in triangle counting~\cite{DBLP:conf/soda/EdenRS20}.
Our work shows the use of arboricity in triangle counting where the graph can be accessed through random edge queries along with other standard queries like \degreeq{}, \neighbourq{} and \edgeexistsq{}. 

Counting triangles has been a well-motivated problem with applications across various domains like optimizing query size in database join problems~\cite{DBLP:conf/soda/Bar-YossefKS02,AtseriasDatabaseJoinJComp,assadi2018simple}, computing clustering coefficients, transitivity ratio~\cite{CAggarwalEvolutionaryNetworkSurvey,LucePerry1949AMO,watts_collective_1998, LeskovecBKT/SIGKDD08/EvolSocNet}, studying structures in web graphs~\cite{Eckmann2001CurvatureOC, DanischBS/WWW18.RealWorldDegeneracy} etc. For a more thorough overview, see the surveys~\cite{Hasan2018TriangleCI,TsourakakisJournalOG}. The parametrization based on arboricity is also of practical interest due to the occurrence of low-arboricity graphs in various real-world scenarios~\cite{DoryGI/BoundedArboricity/PODC22,LowArboricityMatching/FSTTCS2024/Konrad,OnakSSW/LowArboricityMinorFree/ACMTrAlg,GoelGustedtBoundedArboricity,beraCG/BoundedArboricity:LIPIcs.ICALP.2020.11,DanischBS/WWW18.RealWorldDegeneracy,ShinRealWorldKCore/JKInfSyst}.


Formally, we are given a simple, unweighted, undirected graph $\graph = \fbrac{\vertexset,\edgeset}$ with $\size{\vertexset} = \vertexcount$, and $\size{\edgeset} = \edgecount$. The neighbours of $\vertex$ is denoted $\neighbour{\fbrac{\vertex}}$ $=$  $\sbrac{\altvertex|\fbrac{\altvertex,\vertex}\in \edgeset}$ and the degree of the vertex $\vertex$ is denoted $\degree{\vertex}= \size{\neighbour{\fbrac{\vertex}}}$. We are given the following queries to access the graph:
\begin{itemize}
    \item{$\degreeq{(\vertex)}$: } Given a vertex $\vertex$, return $\degree{\vertex}$. 
    \item{$\neighbourq(\vertex,i)$: } Given a vertex $\vertex \in \vertexset$ and $i \in [n]$, return the $i$-th vertex $\altvertex \in \neighbour(\vertex)$ if it exists; otherwise, we get $\perp$.
    \item{$\edgeexistsq(\altvertex,\vertex)$: } Given two vertices $\altvertex,\vertex \in \vertexset$, return $1$ if $(\altvertex,\vertex) \in \edgeset$, and $0$, otherwise.
    \item{$\randedgeq$:} Return an edge $\edge \in \edgeset$ uniformly at random.
\end{itemize}
We will call the $\degreeq$, $\neighbourq$ and $\edgeexistsq$ queries taken together to be \emph{local queries}~\cite{assadi2018simple}. Now, we formally state our problem:

\begin{center}
   \fbox{\begin{minipage}{0.95\linewidth}
\flushleft\textbf{Triangle Estimation: }\\ \centering Given $\approxerror, \confidence$ with $0 < \approxerror, \confidence < 1$, obtain an estimate $\emptriangle$ of the number of triangles $\numtriangle$ in\\the graph $G$, such that $\emptriangle \in \tbrac{\fbrac{1-\approxerror}\numtriangle,\fbrac{1+\approxerror}\numtriangle}$ with probability  at least $1-\confidence$. 
\end{minipage}
} 
\end{center}

From now on, we refer $\emptriangle$ to be an $\appcon$ estimate of $\numtriangle$. One of the graph parameters that has been relevant to quantify the complexity of triangle counting algorithms within the context of RAM model~\cite{DBLP:journals/siamcomp/ChibaN85}, streaming model~\cite{BeraSeshadriStreamingDegeneracy}, distributed model~\cite{JainSeshadriTuranWWW17, SuriVasCursedLastReducerWWW11, FinocchiMapreduceArboricityJExpAlg, LiuSeshadri/PODC/2024/MPCArboricityTriangle} and query model~\cite{DBLP:conf/approx/EdenR18}, is the arboricity of the graph $\graph$, denoted $\arboricity$. The arboricity of a graph can be seen as a measure of the density of the graph.
\begin{definition}[Arboricity$(\arboricity)$]
   The arboricity of a graph $\graph = (\vertexset,\edgeset)$, denoted $\arboricity$, is the minimum number of spanning forests that $\edgeset$ can be partitioned into.
\end{definition}

 Recent works have established various ways to construct lower bounds incorporating the impact of $\approxerror$ and $\confidence$ both in graph property testing~\cite{DBLP:conf/approx/AssadiN22} and other problems~\cite{tetek:LIPIcs.ICALP.2022.107}. Our lower bounds will also involve $\approxerror$ and $\confidence$. To illustrate the dependence of $\approxerror$,  $\confidence$,  and polylog terms on the query complexity of our algorithm, we will use two notations at times in place of $\bigo{\bullet}$ --- $\bigot{\bullet}$ hides polylog factors only but shows dependence on $\approxerror$ and $\confidence$, whereas $\bigoted{\bullet}$ hides both $\approxerror$, $\confidence$ and polylog terms. 

\subsection{Our results}
\label{ssec:ourresults}
In this section, we present the main results of this work. 
\paragraph*{Upper bound}
We show that using random edge queries in addition to local queries, we can obtain an $\appcon$ estimate of $\numtriangle$ where the query complexity is parametrized by arboricity apart from the usual parameters of edge, vertex, the number of triangles, $\approxerror$, and $\confidence$. The work of~\cite{DBLP:journals/siamcomp/ChibaN85} established the role of arboricity in the problem of triangle counting giving an $\bigo{\edgecount\arboricity}$ time algorithm to exactly count the number of triangles $\numtriangle$ in the RAM model. In recent works in the property testing model, ~\cite{Dana_Ron_Triangle_Counting} gave an algorithm that uses {local queries} and requires $\bigoted{\frac{\vertexcount}{\numtriangle^{1/3}}+\min{\sbrac{\edgecount,\frac{\edgecount^{3/2}}{\numtriangle}}}}$ queries to obtain an $\appcon$ estimate $\emptriangle$ of $\numtriangle$. Using \randedgeq{} in addition to the {local} queries, \cite{assadi2018simple} proposed an algorithm that uses $\bigot{\frac{\edgecount^{3/2}\log{\fbrac{1/\confidence}}}{\approxerror^2\numtriangle}}$ queries. Following this, ~\cite{DBLP:conf/soda/EdenRS20} was introduces the parameter of arboricity to triangle counting in the query based setup. They proposed an algorithm that uses $\bigot{\frac{\vertexcount}{\numtriangle^{1/3}}+\frac{\edgecount\arboricity\log^3{\fbrac{1/\confidence}}}{\approxerror^5\numtriangle}}$ {local} queries to obtain an $\appcon$ estimate $\emptriangle$ of $\numtriangle$. Our work follows in the line of the works of ~\cite{Dana_Ron_Triangle_Counting,assadi2018simple,DBLP:conf/soda/EdenRS20}, where \remove{we ask if we can bring in the parameter of arboricity to triangle counting by using \randedgeq{} in addition to other {local} queries (see the last row in Table~\ref{Table: Comparison Table}).  
In this work,} we use \randedgeq{} in addition to \emph{local} queries to obtain an algorithm that uses $\bigot{\frac{\edgecount\arboricity\log{\fbrac{1/\confidence}}}{\approxerror^3\numtriangle}}$ queries (see the last row in Table~\ref{Table: Comparison Table}). Formally,
\begin{theorem}[Upper bound]\label{Theorem: Upper Bound Simplified}
    Given a simple, unweighted and undirected graph $\graph = \fbrac{\vertexset,\edgeset}$ with $\size{\vertexset} = \vertexcount$, $\size{\edgeset} = \edgecount$ and arboricity $\arboricity$, and access to the graph via \degreeq{}, \neighbourq{}, \edgeexistsq{} and \randedgeq{} queries, an $\appcon$ estimate $\emptriangle$ of $\numtriangle$ can be obtained using $\bigot{\frac{\edgecount\arboricity\log{\fbrac{1/\confidence}}}{\approxerror^3\numtriangle}}$ queries.
\end{theorem}

A constant-pass $\bigot{\frac{\edgecount\degeneracy}{\approxerror^4\numtriangle}}$ space streaming algorithm was proposed in~\cite{BeraSeshadriStreamingDegeneracy} for obtaining an $\appcon$-estimate of the number of triangles in a graph, where $\degeneracy$ is the degeneracy of the graph. Our algorithm can be directly extended to a constant pass streaming algorithm using $\bigot{\frac{\edgecount\arboricity}{\approxerror^3\numtriangle}}$ space. By the fact that $\arboricity \leq \degeneracy \leq 2\arboricity$, our algorithm's space complexity is as good as the algorithm of~\cite{BeraSeshadriStreamingDegeneracy} in terms of the graph parameters, and asymptotically better in terms of the approximation factor $\approxerror$.

\paragraph*{Lower bound.}
We also establish a lower bound involving the arboricity and the parameters $\approxerror$ and $\confidence$. Reductions from communication complexity is a well studied method to establish lower bound results in terms of query complexity~\cite{Blais2012,DBLP:conf/approx/EdenR18,Goldreich2020}. Using only \degreeq{} and \neighbourq{} queries, triangle counting requires $\bigomega{\edgecount}$ queries~\cite{GRS11}. Including \edgeexistsq{} with the previous two queries, ~\cite{Dana_Ron_Triangle_Counting} established a lower bound of $\bigomega{\frac{\vertexcount}{\numtriangle^{1/3}}+\frac{\edgecount^{3/2}}{\numtriangle}}$ queries which matches their upper bound up to $\poly{\frac{\log{\fbrac{\vertexcount/\confidence}}}{\approxerror}}$ factors. Later, a simpler proof was proposed for the same bound through reduction from communication complexity problems~\cite{DBLP:conf/approx/EdenR18}. This was subsequently extended to incorporate arboricity with a bound of $\bigomega{\frac{\vertexcount}{\numtriangle^{1/3}}+\frac{\edgecount\arboricity}{\numtriangle}}$ queries~\cite{DBLP:conf/soda/EdenRS20} which again matched their proposed upper bound up to $\poly{\frac{\log{\fbrac{\vertexcount/\confidence}}}{\approxerror}}$ factors. {If \randedgeq{} query is allowed along with all three local queries, the above lower bounds of $\bigomega{\frac{\vertexcount}{\numtriangle^{1/3}}+\frac{\edgecount^{3/2}}{\numtriangle}}$ and $\bigomega{\frac{\vertexcount}{\numtriangle^{1/3}}+\frac{\edgecount\arboricity}{\numtriangle}}$ become $\bigomega{m^{3/2}/T}$ \cite{DBLP:conf/approx/EdenR18} and $\bigomega{m\alpha/T}$ \cite{DBLP:conf/soda/EdenRS20}, respectively.~\footnote{Though \cite{DBLP:conf/soda/EdenRS20} does not mention about $\bigomega{m\alpha/T}$ explicitly, it is not difficult to see it from their construction. Also, this lower bound of $\bigomega{m\alpha/T}$ matches with our upper bound (stated in Theorem~\ref{Theorem: Upper Bound Simplified})  in terms of dependence on $m,T$, and $\alpha$.} As we have mentioned, for general graphs, there is an algorithm that makes $\bigot{\frac{m^{3/2}\log(1/\delta)}{\approxerror^2 T}}$ queries, where each query is either a local query or a \randedgeq{} query \cite{assadi2018simple}.} Recently, in~\cite{DBLP:conf/approx/AssadiN22}, it has been shown that the algorithm of~\cite{assadi2018simple} is optimal (up to $\poly{\log{\fbrac{\vertexcount}}}$ factors) in terms of all parameters including $\approxerror$ and $\confidence$ through a lower bound of $\bigomega{\frac{\edgecount^{3/2}\log{\fbrac{1/\confidence}}}{\approxerror^2\numtriangle}}$ queries for obtaining an $\appcon$ estimate of $\numtriangle$ using {local} and \randedgeq{} queries. However, this does not consider the dependency of the query complexity on arboricity of the graph. Like our upper bound, the lower bound of our work also captures the role of arboricity and the parameters $\approxerror$, and $\confidence$ when the queries allowed are the local and the \randedgeq{} queries. {Moreover, our lower bound combines the parametrization of~\cite{DBLP:conf/approx/AssadiN22} and~\cite{DBLP:conf/soda/EdenRS20} in terms of $\approxerror$, $\confidence$, and $\arboricity$, respectively, to establish an improved lower bound.}

\begin{theorem}[Lower bound]\label{Theorem: Lower Bound Simplified}
    Given a simple, unweighted and undirected graph $\graph = \fbrac{\vertexset,\edgeset}$ with $\size{\vertexset} = \vertexcount$, $\size{\edgeset} = \edgecount$ and arboricity $\arboricity$, and access to the graph via \degreeq{}, \neighbourq{}, \edgeexistsq{} and \randedgeq{} queries, obtaining an $\appcon$ estimate $\emptriangle$ of $\numtriangle$ requires $\bigomega{\frac{\edgecount\arboricity\log{\fbrac{1/\confidence}}}{\approxerror^2\numtriangle}}$ queries.
\end{theorem}
\begin{table}[ht!]
    \centering
    \renewcommand{\arraystretch}{1.6}
    \begin{tabular}{|c|c||c|}
        \hline \hline 
        & \textbf{Without \randedgeq{}} & \textbf{With \randedgeq{}} \\
        \hline \hline 
        \multirow{2}{*}{\textbf{Without Arboricity}} 
            & UB: $\bigoted{\frac{\vertexcount}{\numtriangle^{1/3}} + \frac{\edgecount^{3/2}}{\numtriangle}}$~\cite{Dana_Ron_Triangle_Counting} 
            & UB: $\bigot{\frac{\edgecount^{3/2} \log{\fbrac{1/\confidence}}}{\approxerror^2 \numtriangle}}$~\cite{assadi2018simple} \\
        & LB: $\bigomegat{\frac{\edgecount^{3/2} \log{\fbrac{1/\confidence}}}{\approxerror^2 \numtriangle}}$~\cite{DBLP:conf/approx/AssadiN22} 
            & LB: $\bigomegat{\frac{\edgecount^{3/2} \log{\fbrac{1/\confidence}}}{\approxerror^2 \numtriangle}}$~\cite{DBLP:conf/approx/AssadiN22} \\
        \hline
        \multirow{2}{*}{\textbf{With Arboricity}} 
            & UB: $\bigot{\frac{\vertexcount}{\numtriangle^{1/3}} + \frac{\edgecount \arboricity \log^3{\fbrac{1/\confidence}}}{\approxerror^5 \numtriangle}}$~\cite{DBLP:conf/soda/EdenRS20} 
            & ($\ast$) UB: $\bigot{\frac{\edgecount \arboricity \log{\fbrac{1/\confidence}}}{\approxerror^3 \numtriangle}}$ \\
        & ($\ast$) LB: $\bigomega{\frac{\edgecount \arboricity \log{\fbrac{1/\confidence}}}{\approxerror^2 \numtriangle}}$ 
            & ($\ast$) LB: $\bigomega{\frac{\edgecount \arboricity \log{\fbrac{1/\confidence}}}{\approxerror^2 \numtriangle}}$ \\
        \hline \hline
    \end{tabular}
    \caption{Upper and lower bounds for the Triangle Counting problem. Entries marked with an asterisk ($\ast$) are our results proved in this paper.}
    \label{Table: Comparison Table}
\end{table}

Table~\ref{Table: Comparison Table} describes the upper and lower bounds across different query accesses and parametrization for the triangle counting problem, with our contributions in this work highlighted with an asterisk ($\ast$) mark. Table~\ref{Table: Comparison Table} shows that it is only natural to study the problem of triangle estimation with the use of \randedgeq{} and arboricity as a parameter after the works of~\cite{Dana_Ron_Triangle_Counting,DBLP:conf/soda/EdenRS20,assadi2018simple,DBLP:conf/approx/AssadiN22}. We believe our work also raises some new questions: (a) it opens up the question of quantifying the impact of $\arboricity$ to count other subgraphs (e.g., k-cliques) using \randedgeq{} and other local queries; (b) there remains a small gap of $1/\approxerror$ between our upper and lower bounds that remain to be resolved.


\subsection{Related work in other models}
\label{ssec:related work}
The problem of triangle counting is a problem that has been extensively studied across various models of computation. Here, we discuss the works in RAM and streaming models. 

\noindent\textbf{RAM model.}
The work of~\cite{DBLP:journals/siamcomp/ChibaN85} proposed an $\bigo{\frac{\edgecount\arboricity}{\numtriangle}}$ algorithm in the RAM model that remains the best known till date. Recent works  have shown it to be conditionally optimal under the 3SUM~\cite{DBLP:conf/soda/KopelowitzPP16} and APSP hypotheses~\cite{DBLP:conf/focs/WilliamsX20}. 

\noindent\textbf{Streaming model.}
In the data streaming model, a long line of work~\cite{DBLP:journals/ipl/PaghT12,DBLP:journals/pvldb/PavanTTW13} has culminated in an $\bigot{\frac{\edgecount}{\numtriangle}\fbrac{\edgesensitivity+\sqrt{\vertexsensitivity}}}$-space one-pass streaming algorithm for insertion only model~\cite{KallaugherP17} that has been shown to be optimal through the combined lower bounds proposed in~\cite{BOV_13_Streaming_triangle_counting_hardness,KallaugherP17}. Here,  $\edgesensitivity$ ($\vertexsensitivity$)  denote the maximum number of triangles incident on an edge $\edge$ (a vertex $\vertex$) in the graph $\graph$. There have been works on triangle counting in multi-pass setup as well as in other streaming models (e.g., turnstile, cash-register etc.)~\cite{DBLP:conf/pods/McGregorVV16,DBLP:conf/stacs/BeraC17,KallaugherP17,KallaugherP/soda/2017/HybridTriangleStream,KallaugherKP/focs/2018/SketchingSubgraphCounting,KallaugherMPV/pods/2019/ALModelCycleCount,KallaugherP/stoc/2020/SketchingVsTurnstile,BeraSesh/pods/2020/DegeneracyStream}.

\color{black}

\subsection{Paper organization}
Section~\ref{sec:tech-overview} discusses a broad overview of the techniques we use, and contextualizes it with respect to the relevant prior works. Section~\ref{sec:prelim} sets up the preliminaries for our proofs, including the structural results involving arboricity of a graph. Section~\ref{sec:algo} discusses the query-based algorithm by breaking the discussion into several parts: a weight function (Section~\ref{ssec:weightfunc}), an oracle based algorithm (Sections~\ref{ssec:oracle-algo}) and its implementation (Section~\ref{ssec:oracle-implement}) and the final algorithm (Section~\ref{ssec:final-algo}). The lower bounds are discussed in Section~\ref{sec:lower-bound}. Some of the proofs are deferred to the Appendix. We have also included the full version of the paper in the Appendix.


\section{Technical overview}
\label{sec:tech-overview}
In this section, we give a broad overview of the techniques used in this work. We denote by $\numtriangle_\edge$ the number of triangles incident on the edge $\edge$, and the degree of an edge by $\degree{\edge}$ \remove{the degree of the vertex of smaller degree in the edge,} $=\degree{\fbrac{\altvertex,\vertex}} = \min{\sbrac{\degree{\altvertex},\degree{\vertex}}}$.

\subsection{Upper bound.}


Our starting point is to use \randedgeq{} queries to obtain $\samplededges$, a random sample of edges  of size $\edgesamplesize$ and try to estimate the number of triangles $\numtriangle$ incident on the edges of $\samplededges$.  As the \randedgeq{} query samples edges uniformly at random, the expected number of triangles each sampled edge participates in is $\frac{\numtriangle}{\edgecount}$. Hence, an obvious approach to estimate $\numtriangle$ would be to count the number of triangles $\numtriangle_\edge$ each edge $e$ participates in, and take their average for the sampled edges and scale it appropriately. However, there are two issues with this approach. 

Firstly, counting the number of triangles an edge $\edge$ participates in exactly would require $\bigomega{\degree{\edge}}$ \emph{local} queries. Therefore, we have to sample sufficiently many neighbour vertices for each edge $\edge$ and check if it forms a triangle or not to obtain a good estimate of $\numtriangle_\edge$.

Secondly, even though the expected number of triangles incident on each edge is $\frac{\numtriangle}{\edgecount}$, individual edges $\edge$ can participate in a large number of triangles, i.e., $\numtriangle_\edge$ {can be} $\bigomega{\numtriangle}$. This makes the variance of the estimate too high, and requires $\edgesamplesize$ to be large to obtain a good estimate. To work around this issue, we fix a threshold $\threshold$ based on the arboricity $\arboricity$ of the graph, that has a connection to the number of triangles in the graph. We call edges that participates in more than $\threshold$ triangles \emph{$\threshold$-heavy edges}. Correspondingly, any edge that participates in less than $\threshold$ triangles is called \emph{$\threshold$-light edges}. For now, assume we are given access to an oracle \heavyoracle{} to decide whether an edge is light or heavy. We call the triangles consisting exclusively of heavy edges as \emph{heavy triangles}, and the remaining triangles as \emph{light triangles}. The threshold $\threshold$ is set at an appropriate multiple of $\arboricity$ so that heavy triangles constitute only $\approxerror$ fraction of the triangles. Thus, we can forget about the \emph{$\threshold$-heavy edges} and focus only on counting the light triangles. This brings us to the issue of sampling light triangles uniformly at random. We define a weight function on the edges that assigns each light triangle to one of its constituent light edges. The weight function is defined to be the number of light triangles assigned to the edge, with each heavy edge being assigned a value of $0$. 

Now the algorithm's objective is to estimate this weight function, which cannot be accessed or sampled directly. Observe that there can be multiple assignments of the light triangles to the edges, and consequently multiple weight functions. It suffices for the algorithm to estimate any one of them. In that regard, for all the light edges in our sample $\samplededges$, we maintain a partial non-conflicting assignment. Such an assignment can be extended to a valid assignment, and consequently, a valid weight function. Based on this assignment, we obtain an estimate of the corresponding weight function, which, given an appropriately fixed threshold $\threshold = \Theta\fbrac{\frac{\arboricity}{\approxerror}}$, and setting $\edgesamplesize$ to $\bigot{\frac{\edgecount\arboricity}{\approxerror^3\numtriangle}\log\frac{1}{\confidence}}$ will give an $\appcon$ estimate of $\numtriangle$.

\remove{
such that any edge that participates in more than $\threshold$ triangles is not considered. We call such edges \emph{$\threshold$-heavy edges} and correspondingly any edge that participates in less than $\threshold$ triangles is called \emph{$\threshold$-light edges}.However, observe that this causes the triangles to be sampled non-uniformly, with some triangles not being sampled at all. We call the triangles constituting of exclusively heavy edges \emph{heavy triangles}, and the remaining triangles as \emph{light triangles}.

\blue{We set the $\threshold$ at an appropriate multiple of the arboricity $\arboricity$ of the graph so that heavy triangles constitute only $\approxerror$ fraction of the triangles, and focus on counting only the light triangles. However, this still leaves the issue of light triangles being sampled with different probabilities unsolved. To address this issue, we define a weight function on the edges that assigns each light triangle to one of its constituent light edges. The weight function is defined as number of light triangles assigned to the edge, with each heavy edge being assigned $0$.  }

\blue{Now the algorithm's objective is to estimate this weight function, which can not be accessed or sampled directly. Observe that there can be multiple assignments of the light triangles, and consequently multiple weight functions. It suffices for the algorithm to estimate any one of them. In that regard, for all the light edges in our sample $\samplededges$, we maintain a partial assignment without any conflicting assignment. Such an assignment can be extended to a valid assignment, and consequently, a valid weight function. Hence, we can estimate this weight function through this partial assignment, and in turn obtain an estimate of the number of light triangles.}}

To design the oracle $\heavyoracle{}$ for a given edge $\edge$, we estimate the number of triangles each edge participates in. We can estimate it by sampling each triangle independently through the corresponding neighbouring vertex of $\edge$. The probability of obtaining a triangle for an edge $\edge$ through a random neighbour is $\frac{\numtriangle_\edge}{\degree{\edge}}$. Hence, for the heavy edges where $\numtriangle_\edge \geq \threshold$, $\numtriangle_\edge$ can be estimated using sufficiently small number of queries. For the light edges where $\numtriangle_\edge \leq \threshold$, the probability of obtaining a triangle through a random neighbour can be arbitrarily low. However, note that the lower probability of obtaining triangles allows for high approximation error. We then design a bucketing trick to exploit this trade-off to implement an efficient algorithm to decide whether an edge is heavy or light. For full details, see Section~\ref{sec:algo}.

\subsection{Lower bound} 

Several lower bounds for triangle counting problems based on reductions from the \emph{disjointness} problem in communication complexity~\cite{DBLP:conf/approx/EdenR18, assadi2018simple, DBLP:conf/soda/EdenRS20} are often close to the best possible upper bounds. It is not clear whether this approach can give tight bounds involving $\approxerror$ and  $\confidence$. Recently, Assadi and Nguyen~\cite{DBLP:conf/approx/AssadiN22} showed that any algorithm to find an  $(\approxerror,\confidence)$-estimate of the number of triangles $\numtriangle$ must make at least $\Omega\left(\frac{m^{3/2} \log(1/\confidence)}{\approxerror^2 T}\right)$ queries, where each query is either a local query or a \randedgeq{} query.\footnote{Note that this lower bound matches exactly with the corresponding upper bound in $\cite{assadi2018simple}$.}  To prove the lower bound, they introduced a new query problem called the \emph{Popcount Thresholding Problem} (\ptp{}), and gave reductions from this problem to triangle estimation, instead of using disjointness. We also show a reduction from \ptp{}. In particular, we show that for any value of arboricity $\arboricity$, we can design a graph $\graph$ with arboricity $\arboricity$ such that finding an $\appcon$-estimate in it using $\bigomega{\frac{\edgecount\arboricity\log{\fbrac{1/\confidence}}}{\approxerror^2\numtriangle}}$ queries would violate the lower bound of \ptp{}.

In the $\defptp$ problem, we are given query access to a binary string $\alicestring \in \{0,1\}^\stringlength$ and must decide whether it was generated from one of two similar distributions. In the first distribution, denoted $\ptpdone$, each bit of the string is independently set to 1 with probability $(1 - 2\ptpsep)\ptpprob/\stringlength$; in the second distribution, denoted $\ptpdtwo$, each bit is independently set to 1 with probability $(1 + 2\ptpsep)\ptpprob/\stringlength$. The objective is to distinguish between these two cases by adaptively querying bits of $\alicestring$. \remove{Informally, for suitable choice of parameters $\ptpsep$, $\confidence$, and integer $\stringlength$, the total number of $1$s in $\alicestring$ is a \emph{good} indicator of which distribution it came from. Specifically, }For suitable choice of $\ptpprob$,$\ptpsep$, and $\confidence$, with probability at least $1 - \confidence$, a string from $\ptpdone$ will have fewer than $(1 - \ptpsep)\ptpprob$ $1$s and a string from $\ptpdtwo$ will have more than $(1 + \ptpsep)\ptpprob$ $1$s. Hence, to distinguish whether $\alicestring$ is generated from $\ptpdone$ or $\ptpdtwo$, it suffices to estimate the number of $1$s present in $\alicestring$, i.e. $\norm{\alicestring}_1$ upto multiplicative approximation factor $\ptpsep$. Assadi et al.~\cite{DBLP:conf/approx/AssadiN22} showed that (under suitable choice of parameters) $\Omega\left(\frac{M \log ({1}/{\confidence})}{\gamma^2 k} \right)$ queries to $\alicestring$ are needed to distinguish between the case whether $\alicestring \sim \ptpdone$ and $\alicestring \sim \ptpdtwo$ with probability $1-\confidence$.

We now describe the hard graph instance $G_{\alicestring}(V_{\alicestring}, E_{\alicestring})$ used in the reduction. The core structure of $V_{\alicestring}$ consists of three vertex sets: $A$, $B$, and $\trianglemakerset$, where $|A| = |B| = \Theta(m/\alpha)$ and $|\trianglemakerset| = \Theta(\alpha)$. Every vertex in $A \cup B$ is connected to all vertices in $\trianglemakerset$, forming a bipartite clique between $A \cup B$ and $\trianglemakerset$. The subgraph induced by $A \cup B$ contains $\Theta(m)$ edges, with each vertex having degree at most $O(\alpha)$. We associate each index $i \in [M]$ from the $\defptp$ instance with a potential edge between $A$ and $B$, and include the edge if and only if $x_i = 1$. Thus, the number of edges between $A$ and $B$ equals $\|\alicestring\|_1$, and the total number of triangles in the graph $G_\alicestring$ is $\|\alicestring\|_1 \cdot |\trianglemakerset|$. By choosing parameters such that $M = \Theta(m)$, $k = \Theta(T/\alpha)$, and $\ptpsep = \Theta\fbrac{\approxerror}$, distinguishing whether $\alicestring$ is drawn from $\ptpdone$ or $\ptpdtwo$ becomes equivalent to distinguishing whether the number of triangles in the graph is at most $(1 - \approxerror)T$ or at least $(1 + \approxerror)T$. To complete the reduction, we show that each local query and \randedgeq{} query to $G_{\alicestring}$ can be simulated using only a constant number of queries to $\alicestring$. To enable this, we extend the construction by introducing two additional disjoint vertex sets $A'$ and $B'$, resulting in five disjoint sets: $A$, $B$, $\trianglemakerset$, $A'$, and $B'$. All triangles in the graph are confined to the sets $A$, $B$, and $\trianglemakerset$, while the sets $A'$ and $B'$ are added to ensure that the degree of each vertex in $V_{\alicestring}$ remains identical regardless of whether $\alicestring$ is drawn from $\ptpdone$ or $\ptpdtwo$. With this construction, every allowable query to the graph can be simulated with  $O(1)$ queries to $\alicestring$. For full details, see Section~\ref{sec:lower-bound}.

\subsection{Contextualizing with relevant prior work.} Our work follows as a natural progression to the works of ~\cite{Dana_Ron_Triangle_Counting,assadi2018simple,DBLP:conf/soda/EdenRS20} in the property testing framework. So, let us review the algorithmic techniques at play in them. 


We start by describing a simple randomized algorithm due to Assadi et al.~\cite{assadi2018simple} for estimating the number of triangles in a graph when we have access to queries like \degreeq{}, \neighbourq{}, \edgeexistsq{}, and \randedgeq{}. First, pick a random edge $\{u,v\}$ and let $\degree{u} \leq \degree{v}$. If $\degree{u} \leq \tau_{d}$, pick a random neighbour $w$ of $u$ and check if $\{u,v,w\}$ forms a triangle in the correct order, i.e., in the increasing degree sequence. The probability of sampling a triangle in this case is $ \frac{1}{2m} \cdot \frac{1}{\degree{u}} \cdot \frac{\degree{u}}{\tau_{d}} = \frac{1}{2m\tau_{d}} $. If $\degree{u} > \tau_{d}$, pick a vertex $w$ at random with probability proportional to its degree. This is possible because \randedgeq{} allows sampling vertices proportional to their degrees. The probability of sampling a triangle in this case is $ \frac{1}{2m} \cdot \frac{\degree{w}}{2m} \geq \frac{\tau_{d}}{4m^{2}} $. By setting $ \tau_{d} = \sqrt{2m} $, we can ensure that the algorithm works effectively. This process finds a triangle with probability $ \Omega(T/m^{3/2}) $. Finally, by repeating this process $ \widetilde{O}(m^{3/2}\log (1/\delta)/\approxerror^2 T) $ times, we can estimate the total number of triangles.

The setting $\tau_{d} = \Theta(\sqrt{m})$ seems to be the optimal choice for the general case. We do not see a way to generalize the above idea when parametrized by arboricity.  In bounded arboricity graphs, the structure and degree distribution of the graph may not align with this approach. So, a different strategy would be required for graphs with bounded arboricity to achieve the desired bound of $\widetilde{O}(m\alpha \log (1/\delta)/\approxerror^3 T )$.

 Our approach is more closely aligned with the framework of triangle counting developed by Eden et al.  \cite{Dana_Ron_Triangle_Counting} with access to local queries only.  The algorithm by Eden et al. \cite{Dana_Ron_Triangle_Counting}  starts by picking a small set of vertices at random and for each sampled vertex, the algorithm  tries to estimate the number of triangles it participates in. Some vertices appear in many triangles and are called heavy, while others are called light. To keep the estimate as desired, the algorithm limits the contribution of heavy vertices and focuses mainly on light ones. It gives more weight to triangles that have more light vertices. Since the number of triangles a vertex participates in cannot be obtained directly, the algorithm estimates this by sampling edges incident on the sampled vertices and checking if they close triangles with neighbouring vertices of those edges. The above framework has also been subsequently extended for clique counting in general graphs  \cite{DBLP:conf/soda/EdenRS20} and that in bounded arboricity graphs  \cite{DBLP:conf/soda/EdenRS20}. 

 In summary, the idea of dividing the triangles in terms of their constituent edges or vertices into heavy and light, and dealing with them separately has been a long-standing technique for subgraph counting in sublinear models~\cite{DBLP:conf/pods/McGregorVV16,Dana_Ron_Triangle_Counting, BeraC/stacs/2017/CliqueCountingEA,EdenRS/STOC/2018/SublinearCliques,KallaugherMPV/pods/2019/ALModelCycleCount,DBLP:conf/soda/EdenRS20, BeraSeshadriStreamingDegeneracy}. The algorithmic techniques vary on how to set and leverage this thresholding behaviour for the specific problems and models. Furthermore, allowing \randedgeq{} queries has yielded simpler algorithms for subgraph counting problems~\cite{assadi2018simple}. The usage of \randedgeq{} enables us to design a simple algorithm that uses the structural parameter of arboricity.


The lower bounds for the triangle counting problem in the query model have come from indistinguishability arguments~\cite{Dana_Ron_Triangle_Counting} or reduction from communication problems, particularly \emph{set-disjointness} \cite{DBLP:conf/approx/EdenR18,assadi2018simple,DBLP:conf/soda/EdenRS20} with~\cite{DBLP:conf/soda/EdenRS20} establishing arboricity-sensitive lower bounds. On the other hand, several recent works \cite{BerettaTetek/TALG/2024/BetterSumEstimationWeightedSampling,DBLP:conf/approx/AssadiN22} have established lower bounds for property testing problems through reductions to query problems that are sensitive to $\approxerror$ and $\confidence$. In this work, we extend the $\approxerror,\confidence$-sensitive lower bound of~\cite{DBLP:conf/approx/AssadiN22} to the triangle counting problem for any value of arboricity $\arboricity$.


\section{Preliminaries}
\label{sec:prelim}
\subsection{Notations}
\label{ssec:notation}
The set of positive integers $\{1,2,\ldots,x\}$ is denoted $[x]$. The set of triangles in $\graph$ is denoted $\triangleset$, and individual triangles are denoted $\triangle$. $\fbrac{\vertex,\edge}$ denotes a triangle formed by the vertices $\vertex$ and the endpoints of the edge $\edge$. $\uniform(S)$ denotes an element of $S$ is chosen uniformly at random. 



\subsection{Arboricity and its properties}
\label{ssec:arbor-prop}
As arboricity plays a crucial role in our work, we put together all the structural results that involve arboricity here. Let us restate the definition once more. 
\begin{definition}[Arboricity$(\arboricity)$]
   The arboricity of a graph $\graph = (\vertexset,\edgeset)$, denoted by $\arboricitygraph{G}$, is the minimum number of spanning forests that $\edgeset$ can be partitioned into.
   \label{def:arboricity}
\end{definition}
The arboricity of a graph can be seen as a measure of the density of the graph. $\arboricitygraph{G}$ can be at least $\left\lceil m/(n-1)\right\rceil$. Also, $\arboricitygraph{G} \geq \arboricitygraph{H}$ where $H$ is any subgraph of $G$. We will write $\arboricity$ instead of $\arboricitygraph{G}$ when the underlying graph is understood. We introduce the following result due to~\cite{DBLP:journals/siamcomp/ChibaN85} on the sum of edge degrees over all  edges in the graph.
\begin{lemma}[Sum of $\degree{\edge}$~\cite{DBLP:journals/siamcomp/ChibaN85}]
\label{Lemma: deg(e) sum is m * arboricity}
     Given a graph $\graph = (\vertexset,\edgeset)$ with arboricity $\arboricity$ and $\size{\edgeset} = \edgecount$,  $\sum\limits_{\edge \in \edgeset} \degree{\edge} \leq 2\edgecount\arboricity$.
\end{lemma}

The following result due to~\cite{DBLP:conf/soda/EdenRS20} builds on the work of~\cite{DBLP:journals/siamcomp/ChibaN85} to bound the number of triangles based on the number of edges $\edgecount$ and arboricity $\arboricity$. 
\begin{lemma}[Triangle upper bound ~\cite{DBLP:conf/soda/EdenRS20}]
\label{lemma: arboricity triangle bound}
    Given a graph $\graph = (\vertexset,\edgeset)$ with arboricity $\arboricity$ and $\size{\edgeset} = \edgecount$, the graph $\graph$ has at most $\edgecount\arboricity$ triangles.
\end{lemma}
Note that this upper bound is also tight, i.e., there exists graphs that contain $\edgecount$ edges and $\bigomega{\edgecount\arboricity}$ triangles. Additionally, arboricity $\arboricity$ can be at most $\bigo{\sqrt{\edgecount}}$~\cite{DBLP:conf/soda/EdenRS20}. Thus, all our results can be reformulated by plugging in this upper bound.

\subsection{Chernoff bounds}
We will be using the following variation of the Chernoff bound that bounds the deviation of the sum of independent Poisson trials~\cite{Mitzenmacher_Upfal_2005}.

\begin{lemma}[Multiplicative Chernoff bound]\label{Lemma: Multiplicative Chernoff Bound}
    Given i.i.d. random variables $X_1,X_2,...,X_t$ where $\Pr[X_i = 1] = p$ and $\Pr[X_i = 0] = (1-p)$, define $X = \sum_{i \in [t]} X_i$. Then, we have:
    \begin{align*}
    \Pr[X \leq (1-\approxerror) \Exp\tbrac{X}] &\leq \exp{\fbrac{-\frac{\Exp\tbrac{X}\approxerror^2}{3}}} & 0 \leq \approxerror <1\\
    \Pr[X \geq (1+\approxerror) \Exp\tbrac{X}] &\leq \exp{\fbrac{-\frac{\approxerror^2\Exp\tbrac{X}}{2+\approxerror}}} & 0 \leq \approxerror 
    \end{align*}
\end{lemma}


\section{Algorithm}
\label{sec:algo}

This section describes the algorithm and its related concepts. Section~\ref{ssec:weightfunc} formally defines the weight function for edges and associated structural results. Section~\ref{ssec:oracle-algo} describes the algorithm assuming access to an idealized oracle that can decide if an edge is heavy or light (based on whether the edge has many or few incident triangles). Section~\ref{ssec:oracle-implement} describes how to actually implement this oracle within the problem setup. Section~\ref{ssec:final-algo} puts everything together to develop the final algorithm.

\remove{This section describes the algorithm starting from the setting of an idealized oracle that can decide an edge to be heavy or light (based on the edge having many or less triangles incident on it) based on an idealized weight function over the edges. This weight function is developed in Section~\ref{ssec:weightfunc} starting from the notion of heavy and light edges. Section~\ref{ssec:oracle-algo} discusses about how the idealized oracle can develop an estimate of the weight function. The implementation of the idealized oracle is discussed in Section~\ref{ssec:oracle-implement}. Section~\ref{ssec:final-algo} puts everything together to develop the final algorithm.}



\subsection{Weight function}
\label{ssec:weightfunc}
In this section, we formalize the ideas of heavy and light edges and weight function for the edges. 
\paragraph*{Heavy and light edges and triangles.} First, we define heavy and light edges and correspondingly, heavy and light triangles.
\begin{definition}[$\threshold$-heavy and $\threshold$-light edges]\label{Definition: Heavy and Light Edges}
    An edge $\edge \in \edgeset$ is defined to be a $\threshold$-heavy (resp. $\threshold$-light) edge if it participates in more than $\threshold$ (resp. $\leq \threshold$) triangles.
\end{definition}
\begin{definition}[$\threshold$-heavy and $\threshold$-light triangles]\label{Definition: Heavy and Light Triangles}
    A triangle $\triangle \in \triangleset$ is called a $\threshold$-heavy triangle if all its three edges are $\threshold$-heavy edges. A triangle that is not $\threshold$-heavy is a $\threshold$-light triangle.
\end{definition}


We denote by $\lighttriangles{\threshold}$ and $\heavytriangles{\threshold}$ 
the number of $\threshold$-light and $\threshold$-heavy triangles in the graph $\graph$, respectively. 
The following lemma bounds the number of $\threshold$-heavy triangles in a graph using the parameter of arboricity.
\begin{lemma}[Upper bound on $\heavytriangles{\threshold}$]\label{Lemma: Upper Bound on Heavy Triangles}
Given a graph $\graph = (\vertexset,\edgeset)$ with $\numtriangle$ triangles, the number of $\threshold$-heavy triangles is at most $\frac{3\numtriangle\arboricity}{\threshold}$.
    \begin{proof}
        Note that the graph $\graph = (\vertexset,\edgeset)$ has $\numtriangle$ triangles containing at most $3\numtriangle$ edges. By Definition~\ref{Definition: Heavy and Light Triangles}, $\threshold$-heavy triangles have all three of their edges to be $\threshold$-heavy edges, $\threshold$ or more triangles incident on each of them. Hence, the number of $\threshold$-heavy edges in $\graph$ is at most $\frac{3\numtriangle}{\threshold}$ .

        Now consider the subgraph $H = (V_H, E_H)$ of $\graph$ induced by the $\threshold$-heavy edges in $\graph$. We know, $\edgeset_H \leq \frac{3\numtriangle}{\threshold}$. Also, $\arboricitygraph{H} \leq \arboricitygraph{\graph} = \arboricity$ (see Section~\ref{ssec:arbor-prop}). Hence, by Lemma~\ref{lemma: arboricity triangle bound}, we know that $H$ contains at most $\frac{3\numtriangle\arboricity}{\threshold}$ triangles.
    \end{proof}
\end{lemma}
The following corollary follows from Lemma~\ref{Lemma: Upper Bound on Heavy Triangles} and the fact that $\numtriangle = \lighttriangles{\threshold} + \heavytriangles{\threshold}$. 
\begin{corollary}[Lower bound on $\lighttriangles{\threshold}$]\label{Corollary: Lower Bound on Light Triangles}
    Given a graph $\graph = (\vertexset,\edgeset)$, there are at least $(1-\frac{3\arboricity}{\threshold}) \numtriangle$  $\threshold$-light triangles.
\end{corollary}

\paragraph*{Weight function for the edges.}
Now, we define a weight function for the edges. As a triangle is incident to multiple edges, the objective of a weight function is to charge each of the light triangles to exactly one of its participating light edges. To ensure this, we do not charge any light triangle to any heavy edge that it may be incident on. Each light edge $\edge$ can be charged by at most $\numtriangle_\edge$ triangles, i.e., all the triangles the edge $\edge$ participates in. To avoid over-counting, we ensure that the sum of the weight function over all edges is equal to the number of light triangles, $\lighttriangles{\threshold}$. $\fbrac{\triangle,\edge}$ denotes that the triangle $\triangle$ is charged to the edge $\edge$.


\begin{definition}[Triangle weight function $(\weightfunc)$]\label{Definition: Weight Function}
    We define a function $\func{\weightfunc}{\edgeset}{\Nat}$ to be a triangle weight function if it satisfies the following two conditions:
\[
\begin{array}{llll}
\mbox{Condition (1)}: & \weightfunc(\edge) & \leq & 
\left\{
\begin{array}{ll}
\numtriangle_\edge & \text{if $\edge$ is a light edge}\\
0 &\text{if $\edge$ is a heavy edge} 
\end{array} 
\right. 
\\
\mbox{Condition (2)}: & \sum\limits_{\edge \in \edgeset} \weightfunc(\edge) & = & \lighttriangles{\threshold} 
\end{array}
\]
\end{definition}

Observe that there are multiple such weight functions (e.g., a triangle with 3 light edges can be assigned to any one of these 3 edges). Henceforth, we denote the set consisting of valid triangle weight functions by $\weightfamily$. We now state a property of valid weight functions. Note that this property is true for any $\weightfunc \in \weightfamily$.

\begin{lemma}[Expectation of $\weightfunc(\edge)$]\label{Lemma: Weight Function Expectation}

    Consider any triangle weight function $\weightfunc \in \weightfamily$. If we select an edge $e \in E$ uniformly at random, then the expected value of $w(e)$ is $\frac{\lighttriangles{\threshold}}{\edgecount}$, i.e., $\Exp_{e \sim \uniform(\edgeset)} \weightfunc(\edge) = \frac{\lighttriangles{\threshold}}{\edgecount}$.

    \begin{proof}
        Given that the edges have been chosen uniformly at random and the condition of the triangle weight function (\Cref{Definition: Weight Function}), we have:
        \begin{align*}
            \Exp_{\edge \sim \uniform(\edgeset)} \weightfunc(\edge)  ~=  \sum_{\edge \in \uniform(\edgeset)} \frac{1}{\edgecount} \weightfunc(\edge)   &~=  ~~\frac{1}{\edgecount} \sum_{\edge \in \edgeset} \weightfunc(\edge)  ~=  ~~\frac{\lighttriangles{\threshold}}{\edgecount}. &&\qedhere
        \end{align*}
    \end{proof}
\end{lemma}


\subsection{Oracle based algorithm}
\label{ssec:oracle-algo}
Our algorithm, due to its usage of the \emph{triangle weight function}, requires knowledge of whether an edge is \emph{heavy} or \emph{light}. The problem of deciding whether an edge is heavy is non-trivial as it directly relates to number of triangles the edge participates in. In this section, we assume black-box access to an oracle $\exactheavyoracle{}$ that helps us determine if an edge is heavy. 
\begin{align*}
    \exactheavyoracle(\edge,\arboricity,\approxerror) &=\begin{dcases}
        1 &\text{if edge $e$ is a $\frac{\upperthreshold\arboricity}{\approxerror}$-heavy edge, i.e., $\threshold = \frac{\upperthreshold\arboricity}{\approxerror}$}\\
        0 &\text{if edge $e$ is a $\frac{\lowerthreshold\arboricity}{\approxerror}$-light edge, i.e., $\threshold = \frac{\lowerthreshold\arboricity}{\approxerror}$}
    \end{dcases}
\end{align*}
Here $h$ and $l$ $(h > l)$ are constants to be determined later. We will discuss an efficient 
implementation of this oracle later. 

Let us first consider the case where we are given an oracle that given an edge $\edge$, returns the exact value of a weight function, $\weightfunc(\edge)$. Given we can sample edges uniformly at random through the \randedgeq{} query, we can compute $\Exp_{\edge \sim \uniform(\edgeset)}\tbrac{\weightfunc(\edge)} = \lighttriangles{\frac{\lowerthreshold\arboricity}{\approxerror}}$ using this oracle on the sampled edges. However, no such oracle exists in our model. Hence, we try to simulate one through an empirical estimate of the weight function, $\empweightfunc(\edge)$. Recall the fact that there are many valid weight functions, each being characterized by every light triangle being charged to one of its constituent light edges. Our algorithm will work if the empirical weight function $\empweightfunc$ estimates any one of these weight functions. We achieve this by ensuring that a triangle is not sampled through more than one edge (see Line~\ref{Line: Remove Duplicate Triangles} of Algorithm~\ref{Algorithm: Random Edge Arboricity Triangle Counting Oracle Triangle Estimate}). Thus, the assignments that we consider can be extended to a valid weight function $\weightfunc$, and our algorithm can be thought of as estimating this weight function through $\empweightfunc$. We develop our initial algorithm (Algorithm~\ref{Algorithm: Random Edge Arboricity Triangle Counting Oracle Triangle Estimate}) assuming $\esttriangle$, an estimate of $\numtriangle$ satisfying $\esttriangle \leq 2\numtriangle$.

In Algorithm~\ref{Algorithm: Random Edge Arboricity Triangle Counting Oracle Triangle Estimate}, we assume that we know $m$ exactly. In fact, our algorithm and analysis work even if we have a constant factor approximation of $m$. This can be achieved by using $\bigoted{1}$  queries \cite{assadi2018simple}.

We now state and analyse the algorithm.

\begin{algorithm}[ht!]
    \caption{Triangle Counting Algorithm - with access to \heavyoracle{}, and knowledge of $\esttriangle$}\label{Algorithm: Random Edge Arboricity Triangle Counting Oracle Triangle Estimate}
    \begin{algorithmic}[1]
        \Require \degreeq{}, \neighbourq{}, \edgeexistsq{}, and \randedgeq{} query access to a graph $\graph$. Parameters $\esttriangle, \arboricity$, $\approxerror$, $\edgecount$ and oracle access to \exactheavyoracle{} with threshold constants $\lowerthreshold,\upperthreshold$
        \State $\edgesamplesize \gets \constant(1+\upperthreshold)\approxerror^{-3}(\edgecount\arboricity/\esttriangle)\log\vertexcount$ 
        \State $\samplededges \gets \emptyset$ \Comment{$\samplededges$ is the set of random edges sampled. $\size{\samplededges}$ grows up to $\edgesamplesize$}
        \State $\sampledtriangles \gets \emptyset$ \Comment{$\sampledtriangles$ is the set of light triangles sampled through the edges in $\samplededges$}
        \For{$i \in [\edgesamplesize]$}
            \State $\edge_i \gets \randedgeq{}$
            \State $\samplededges \gets \samplededges \cup \edge_i$
            \State Let $\edge_i = (\vertex_i, x)$ where $\degree{\vertex_i} < \degree{x}$
            \remove{Let $\vertex_i$ be the endpoint of $\edge_i$ with smaller degree, and $x$ be the endpoint of $\edge_i$ that is not $\vertex_i$.} \Comment{Requires two \degreeq{} queries}
            \If{($\exactheavyoracle(\edge_i,\arboricity,\approxerror) = 0$)} \Comment{$\edge_i$ is a $\frac{\lowerthreshold\arboricity}{\approxerror}$-light edge} 
            \label{line: heavyoracle call}
                \State $\querycount_{\edge_i} \gets 0$ \Comment{$\querycount_\edge$ denotes the number of queries for each edge $\edge$}
                \If{$\degree{\vertex_i} \leq \arboricity$}
                    \State set $\querycount_{\edge_i} \gets 1$ with probability $\frac{\degree{\vertex_i}}{\arboricity}$
                \Else 
                    \State set $\querycount_{\edge_i} \gets \ceil{\frac{\degree{\vertex_i}}{\arboricity}}$
                \EndIf
            \For{$j \in [\querycount_{\edge_i}]$}
                \State Choose $k \gets \uniform\fbrac{\sbrac{1,2,...,\degree{\vertex_i}}}$
                \State $\altvertex \gets \neighbourq{\fbrac{\vertex_i,k}}$
                \If{($\edgeexistsq{\fbrac{\altvertex,x}} = 1$ and $\triangle = (\altvertex,\edge_i)$ is not in $\sampledtriangles$ through another edge $\edge'$)}
                \label{Line: Remove Duplicate Triangles} 
                    \State $\sampledtriangles \gets \sampledtriangles \cup \triangle$ \\
                    \Comment{A triangle $\triangle$ may occur in $\sampledtriangles$ multiple times, but each time it will be through the same edge $\edge$}
                \EndIf
            \EndFor
            \Else
                \State $\querycount_{\edge_i} \gets 0$
            \EndIf
        \EndFor

        \For{$i \in [\edgesamplesize]$}
            \If{$\querycount_{\edge_i} > 0$}
                \State $\empweightfunc(\edge_i) = \frac{1}{\querycount_{\edge_i}} \sum_{(\triangle,\edge_i) \in \sampledtriangles} \max\fbrac{\arboricity,\degree{\edge_i}}$
            \Else
                \State $\empweightfunc\fbrac{\edge_i} = 0$
            \EndIf
        \EndFor
        \State \Return $\emptriangle = \frac{\edgecount}{\edgesamplesize}\sum_{i \in \edgesamplesize} \empweightfunc(\edge_i)$
    \end{algorithmic}
\end{algorithm}

\begin{restatable}[Expectation of Algorithm~\ref{Algorithm: Random Edge Arboricity Triangle Counting Oracle Triangle Estimate}]{lemma}{weightfunclemma}\label{lemma: E[Y_I] Weight Func Algo}
    Algorithm~\ref{Algorithm: Random Edge Arboricity Triangle Counting Oracle Triangle Estimate} ensures that $\Exp\tbrac{\empweightfunc(\edge_i)} = \lighttriangles{\frac{\lowerthreshold\arboricity}{\approxerror}}/\edgecount $, and $\Exp\tbrac{\emptriangle} = \lighttriangles{\frac{\lowerthreshold\arboricity}{\approxerror}}$.
\end{restatable}

\begin{proof}
    Let $\sE_i$ be the event that the edge $\edge$ is chosen as the $i$-th edge sample. We proceed with the proof by considering two different cases: $\degree{\edge} < \arboricity$ and $\degree{\edge} \geq \arboricity$. 
    
    
    {\bf Case I }($\degree{\edge} < \arboricity$):
    When $\degree{\edge} < \arboricity$, $\querycount_\edge$ is set to $0$ with probability $1 - \degree{\edge}/\arboricity$ and \complain{1} with probability $\degree{\edge}/\arboricity$, and $\empweightfunc(\edge)$ is evaluated through a single query. Thus, \begin{align}
        \Exp[\empweightfunc(\edge_i)|\sE_i] &= \frac{\degree{\edge}}{\arboricity} \sum_{k \in [\weightfunc(\edge)]} \frac{1}{\degree{\edge}} \arboricity + \fbrac{1 - \frac{\degree{\edge}}{\arboricity}}\cdot0 = \weightfunc(\edge)\label{Eq: E[what(e)] low degree edge}
    \end{align}

     
    {\bf Case II }($\degree{\edge} \geq \arboricity$):
    On the other hand, when $\degree{\edge} \geq \arboricity$, let $Z_j, j \in [\querycount_\edge]$ be the contribution of each $\querycount_\edge$ queries made for the edge $\edge$ to the weight function estimate. Then,
    \begin{align}
        \Exp[Z_j|\sE_i] &= \sum_{k \in [\weightfunc(\edge)]} \frac{1}{\degree{\edge}} \degree{\edge} = \weightfunc(\edge)\label{Eq: E[Z_j] random edge arboricity}
    \end{align}
    Correspondingly, we have by linearity of expectation and Equation~\eqref{Eq: E[Z_j] random edge arboricity}:
    \begin{align}
        \Exp\tbrac{\empweightfunc(\edge_i)|\sE_i} = \frac{1}{\querycount_\edge}\sum_{j \in [\querycount_{\edge}]} \Exp\tbrac{Z_j|\sE_i} = \Exp\tbrac{Z_j|\sE_i} = \weightfunc(\edge)\label{Eq: E[what(e)] high degree edge}
    \end{align}

    Given that we draw each edge $\edge \in \samplededges$ uniformly at random, we now have by Equations~\eqref{Eq: E[what(e)] low degree edge},~\eqref{Eq: E[what(e)] high degree edge}, and Lemma~\ref{Lemma: Weight Function Expectation}:
    \begin{align*}
        \Exp\tbrac{\empweightfunc(\edge_i)} = \Exp_{\edge \sim \uniform\fbrac{\edgeset}} \weightfunc\fbrac{\edge} = \lighttriangles{\frac{\lowerthreshold\arboricity}{\approxerror}}/\edgecount 
    \end{align*}
    By linearity of expectations, we obtain
    \begin{align*}
        \Exp\tbrac{\emptriangle} = \Exp\tbrac{\frac{\edgecount}{\edgesamplesize}\sum_{i \in \edgesamplesize} \empweightfunc(\edge_i)} = \lighttriangles{\frac{\lowerthreshold\arboricity}{\approxerror}} &\qedhere
    \end{align*}

    \end{proof}

Note that Lemma~\ref{lemma: E[Y_I] Weight Func Algo} holds irrespective of whether the assumption $\esttriangle \leq 2\numtriangle$ is satisfied or not. Next, we turn our attention to the variance of $\empweightfunc(\edge)$.

\begin{restatable}[Variance of Algorithm~\ref{Algorithm: Random Edge Arboricity Triangle Counting Oracle Triangle Estimate}]{lemma}{varwtfunc}\label{lemma: Var[Y_i] Weight Func Algo}
    Algorithm~\ref{Algorithm: Random Edge Arboricity Triangle Counting Oracle Triangle Estimate} ensures that $\Var[\empweightfunc(\edge_i)] \leq \frac{(1+\upperthreshold)\arboricity}{\approxerror} \cdot \Exp_{\edge \sim \uniform\fbrac{\edgeset}}[\weightfunc(\edge)]$.
\end{restatable}

\begin{proof}
        Again, we consider two different cases as in the proof of Lemma~\ref{lemma: E[Y_I] Weight Func Algo}. 
        
      
      {\bf Case I }($\degree{\edge} < \arboricity$):
        When $\degree{\edge} < \arboricity$, Algorithm~\ref{Algorithm: Random Edge Arboricity Triangle Counting Oracle Triangle Estimate} makes at most $1$ query as $\querycount \leq 1$. As $\empweightfunc(\edge) \leq \arboricity$, we have 
        \begin{align}
            \Var[\empweightfunc(\edge_i)|\sE_i] \leq \Exp[\empweightfunc(\edge_i)^2|\sE_i] \leq \arboricity\Exp\tbrac{\empweightfunc\fbrac{\edge_i}|\sE_i} \label{Eq: Var[Y_i,S, Small Degree Edge] random edge arboricity}
        \end{align}
        
        \remove{
        $$\leq \frac{\upperthreshold\arboricity}{\approxerror}\Exp[\empweightfunc\fbrac{\edge_i}|\sE_i] $$
        \begin{align*}
            \Var[\empweightfunc(\edge_i)] &\leq \Exp[\empweightfunc(\edge_i)^2]\\
                        &\leq \arboricity\Exp\tbrac{\empweightfunc\fbrac{\edge}}&\text{As $\empweightfunc(\edge) \leq \arboricity$}\\
                        &\leq \frac{\upperthreshold\arboricity}{\approxerror}\Exp[\empweightfunc]
        \end{align*}
        }
        
        {\bf Case II }($\degree{\edge} \geq \arboricity$):
        Now we consider the case $\degree{\edge} \geq \arboricity$ which is more involved as $\querycount \geq 1$. We first estimate the variance of each $Z_j$ individually as defined in Lemma~\ref{lemma: E[Y_I] Weight Func Algo}. We condition on the event that edge $\edge$ is chosen as the $i$-th edge sample, denoted by $\sE_i$. If $\exactheavyoracle(\edge,\arboricity,\approxerror) = 1$, then $Var[Z_j] = 0, \forall j \in [\querycount]$. Hence, we condition on the event that $\exactheavyoracle(\edge,\arboricity,\approxerror) = 0$ from now on. As $Z_j \leq \degree{\edge}$,
        \begin{align}
             \Var[Z_j|\sE_i]  \leq \Exp[Z_j^2|\sE_i] 
                                \leq \degree{\edge_i}\Exp[Z_j|\sE_i] \label{Eq: Var[Z_j]}
        \end{align}
        \remove{
        \begin{align}
            \nonumber Var[Z_j|\sE_i]    &\leq \Exp[Z_j^2|\sE_i]&\\
                                &\leq \degree{\edge}\Exp[Z_j|\sE_i] &\text{$Z_j \leq \degree{\edge}$}\label{Eq: Var[Z_j]}
        \end{align}
        }
        After having bounded the variance of the contribution of each query, we now obtain the variance of the weight estimate of the edge.
        \begin{align}
          \nonumber  \Var[\empweightfunc(\edge_i)|\sE_i]   &=\Var\tbrac{\frac{1}{\querycount_{\edge_i}}\sum_{j \in \querycount_{\edge_i}} Z_j|\sE_i}&\\
          \nonumber                      &=\frac{1}{\querycount_{\edge_i}^2}\sum_{j \in \querycount_{\edge_i}}\Var\left[ Z_j|\sE_i \right]&\text{($Z_j$ are i.i.d. given $\sE_i$)}\\
          \nonumber                      &=\frac{\degree{\edge_i}}{\querycount_{\edge_i}}\sum_{j \in \querycount_{\edge_i}}\frac{1}{\querycount_{\edge_i}}\Exp[ Z_j|\sE_i ]&\text{(by Equation~\eqref{Eq: Var[Z_j]} and linearity of expectation)}\\
          &\leq \arboricity \Exp[\empweightfunc(\edge_i)|\sE_i]&\text{ $\left(\mbox{as }  \querycount_{\edge_i} = \ceil{\frac{\degree{\edge}}{\arboricity}}\right)$}\label{Eq: Var[Y_i,S,High Degree Edge] random edge arboricity}
\end{align}

        Now, we remove the conditioning on $\sE_i$ using the law of total variance:
        \begin{align*}
            \Var[\empweightfunc(\edge_i)] &= \Exp_{\edge_i}[\Var[\empweightfunc(\edge_i)|\sE_i]] + \Var_{\edge_i}[\Exp[\empweightfunc(\edge_i)|\sE_i]] &\text{(by the law of total variance)}\\
            &\leq \Exp_{\edge_i}[\arboricity \Exp[\empweightfunc(\edge_i)|\sE_i]] + \Exp_{\edge_i} [\Exp[\empweightfunc(\edge_i)|\sE_i]^2]&\text{(by Equations~\eqref{Eq: Var[Y_i,S, Small Degree Edge] random edge arboricity} and~\eqref{Eq: Var[Y_i,S,High Degree Edge] random edge arboricity})}\\
            &\leq \arboricity \Exp_{\edge_i}[\Exp[\empweightfunc(\edge_i)|\sE_i]] + \frac{\upperthreshold\arboricity}{\approxerror} \cdot \Exp_{\edge_i} [\Exp[\empweightfunc(\edge_i)|\sE_i]]&\text{(as $\exactheavyoracle(\edge) = 0$, $\frac{\upperthreshold\arboricity}{\approxerror}\geq |T_e|$)}\\
            &\leq \frac{(1+\upperthreshold)\arboricity}{\approxerror} \cdot \Exp[\weightfunc(\edge)] &\qedhere
        \end{align*}
    \end{proof}

\begin{theorem}[Performance guarantees of Algorithm~\ref{Algorithm: Random Edge Arboricity Triangle Counting Oracle Triangle Estimate} ]\label{Theorem: Oracle Triangle Estimate ALgo Works}
    Algorithm~\ref{Algorithm: Random Edge Arboricity Triangle Counting Oracle Triangle Estimate} makes $\constant(1+\upperthreshold)\approxerror^{-3} (\edgecount\arboricity/\esttriangle)\log\vertexcount$ queries, and $\constant(1+\upperthreshold)\approxerror^{-3} (\edgecount\arboricity/\esttriangle)$ $\log\vertexcount$ calls to \exactheavyoracle{}, and given $\esttriangle$ satisfying the assumption that $\esttriangle \leq 2\numtriangle$,  returns $\emptriangle$ such that $\Pr(|\emptriangle-\lighttriangles{\frac{\lowerthreshold\arboricity}{\approxerror}}|\geq \approxerror\lighttriangles{\frac{\lowerthreshold\arboricity}{\approxerror}}) \leq \frac{1}{\constant\log \vertexcount}$.
\end{theorem}

\begin{proof}
        The algorithm calls \exactheavyoracle{} for each of the $\edgesamplesize$ edges in $\samplededges$, resulting in a total of $\constant(1+\upperthreshold)\approxerror^{-3}(\log\vertexcount)(\edgecount\arboricity/\esttriangle)$ calls. 

        The algorithm makes $\edgesamplesize$ \randedgeq{} queries, $2\edgesamplesize$ \degreeq{} queries, and $2\querycount_\edge$ \neighbourq{} query for each edge $\edge \in \samplededges$. All edges in $\samplededges$ are sampled uniformly at random from $\edgeset$. Hence, the expected number of \neighbourq{} queries made are: 
        
        $$\Exp_{\edge \sim \uniform\fbrac{\edgeset}} \left[ \ceil{\frac{\degree{\edge}}{\arboricity}} \right] = \sum_{e \in E} \frac{1}{m} \ceil{\frac{\degree{\edge}}{\arboricity}} \leq \frac{1}{m} \sum_{e \in E} 1 + \frac{\degree{\edge}}{\arboricity} = 1 + \frac{1}{m} \sum_{e \in E} \frac{\degree{\edge}}{\arboricity} \leq 1 + \frac{2m\arboricity}{m\arboricity} = 3$$
        The last inequality in the above step follows from Lemma~\ref{Lemma: deg(e) sum is m * arboricity}.
        \remove{
        \begin{align*}
            &\Exp_{\edge \sim \uniform\fbrac{\edgeset}} \ceil{\frac{\degree{\edge_i}}{\arboricity}}\\
            =&\sum_{e \in E} \frac{1}{m} \ceil{\frac{\degree{\edge}}{\arboricity}}\\
            \leq&\frac{1}{m} \sum_{e \in E} 1 + \frac{\degree{\edge}}{\arboricity}\\
            =&1 + \frac{1}{m} \sum_{e \in E} \frac{\degree{\edge}}{\arboricity}\\
            \leq& 1 + \frac{2m\arboricity}{m\arboricity}&\text{By Lemma~\ref{lemma: arboricity triangle bound}}\\
            =& 3
        \end{align*}
        }
        Hence, the algorithm makes at most $9\edgesamplesize = \constant(1+\upperthreshold)\approxerror^{-3}(\edgecount\arboricity/\esttriangle) \log \vertexcount$ queries in expectation. Here the inequality is due to the fact that for a heavy edge $\edge$ in $\samplededges$, $\querycount_\edge = 0$. Furthermore, we bound the variance of the estimate $\emptriangle$ as:

        $$\Var[\emptriangle] = \Var\left[\frac{\edgecount}{\edgesamplesize} \sum_{i \in s} \empweightfunc(\edge_i)\right] \leq \frac{(1+\upperthreshold)\arboricity\edgecount^2\Exp[\weightfunc(\edge)]}{\approxerror \edgesamplesize} \leq \frac{(1+\upperthreshold)\edgecount\arboricity\Exp[\emptriangle]}{\approxerror \edgesamplesize}.$$
        The last two steps follow from Lemma~\ref{lemma: Var[Y_i] Weight Func Algo} and the fact that $\Exp\tbrac{\emptriangle} = \edgecount\Exp\tbrac{\weightfunc(\edge)}$.
        We now use Chebyshev's inequality on $\emptriangle$:
        \begin{align*}
            \Pr\left(|\emptriangle - \Exp[\emptriangle]| \leq \approxerror \Exp[\emptriangle]\right) 
                    &\leq \frac{\Var[\emptriangle]}{\approxerror^2\Exp[\emptriangle]^2}   &\text{(Chebyshev's inequality)}\\
                    &\leq \frac{(1+\upperthreshold)\edgecount\arboricity}{\approxerror^3\edgesamplesize\Exp[\emptriangle]}&\text{(by Lemma~\ref{lemma: Var[Y_i] Weight Func Algo})}\\
                    &\leq \frac{(1+\upperthreshold)\edgecount\arboricity}{\approxerror^3\cdot 4\constant\fbrac{1+\upperthreshold}\approxerror^{-3}(\edgecount\arboricity/\esttriangle)\log\vertexcount \cdot \lighttriangles{\frac{\lowerthreshold\arboricity}{\approxerror}}} &\text{(by Lemma~\ref{lemma: E[Y_I] Weight Func Algo})}\\
                    & \leq \frac{1}{\constant\log \vertexcount} &\text{(as $\lighttriangles{\frac{\lowerthreshold\arboricity}{\approxerror}} > \numtriangle/2 \geq \esttriangle/4$)}
        \end{align*}
    where the last inequality uses the fact $\numtriangle/2 \geq \esttriangle/4$ due to the assumption that $\esttriangle \leq 2\numtriangle$.
    \end{proof}



\subsection{Implementing the oracle}
\label{ssec:oracle-implement}
Rather than the exact oracle (\exactheavyoracle) that we assumed in Algorithm~\ref{Algorithm: Random Edge Arboricity Triangle Counting Oracle Triangle Estimate}, we would design an oracle (called \heavyoracle) that given arboricity $\arboricity$, and parameters $\approxerror$ and $\confidence$, accepts edges participating in at most $\frac{\arboricity}{2\approxerror}$ triangles as light, and rejects edges participating in at least $\frac{2\arboricity}{\approxerror}$ triangles as heavy, with probability at least $1 - \confidence$. The algorithm works by estimating the number of triangles each edge participates in. However, in this case, each \neighbourq{} and \edgeexistsq{} queries generate i.i.d. random variables denoting the existence of triangle through that neighbouring vertex. Thus, our analysis uses the multiplicative Chernoff bound to obtain high probability guarantees for each individual edge.

\begin{algorithm}[ht!]
    \caption{\heavyoracle($\edge$, $\arboricity$, $\approxerror$, $\confidence$)}\label{Algorithm: Heavy Oracle}
    \begin{algorithmic}[1]
        \Require \degreeq{}, \neighbourq{}, \edgeexistsq{}, and \randedgeq{} query access to a graph $\graph$
        \State $\querycount \gets \ceil{\frac{16\approxerror ~ \degree{\edge}}{\arboricity}\log\fbrac{\frac{1}{\confidence}}}$ \Comment{$\querycount$ denotes the number of queries for each edge $\edge$}
        \State Let $\edge = (\altvertex, x)$ where $\degree{\altvertex} < \degree{x}$
        \Comment{Requires 2 \degreeq{} queries}
        \State $Y \gets 0$
        \For{$i \in [\querycount]$}
            \State Choose $k \gets \uniform\fbrac{\sbrac{1,2,...,\degree{\altvertex}}}$
            \State $\vertex_i \gets \neighbourq(\altvertex,k)$ 
            \Comment{Requires 1 \neighbourq{} query}
            \If{$\edgeexistsq{\fbrac{\vertex_i,x}} = 1$} \Comment{Requires 1 \edgeexistsq{} query}
                \State $Y_i \gets 1$ 
                \Comment{Found triangle $\fbrac{\vertex_i,\edge}$ }
            \Else
                \State $Y_i \gets 0$
            \EndIf
            \State $Y  \gets Y + Y_i$
        \EndFor
        \State $Y \gets \frac{Y}{\querycount}$
        \If{$Y \geq \frac{\arboricity}{\approxerror\cdot\degree{\edge}}$} 
            \State \Return 1
        \Else
            \State \Return 0
        \EndIf
    \end{algorithmic}
\end{algorithm}
\begin{restatable}[Estimation guarantee of \heavyoracle{}]{lemma}{heavyoraclect}\label{Lemma: Heavy Oracle Algorithm Correctness}
    The algorithm $\heavyoracle{\fbrac{\edge,\arboricity,\approxerror,\confidence}}$ satisfies the following properties with probability at least $1-\confidence$: (i) returns $1$ edge $\edge$ if it is $\frac{2\arboricity}{\approxerror}$-heavy; (ii) returns $0$ edge $\edge$ if it is $\frac{\arboricity}{2\approxerror}$-light.
   \remove{
    \begin{itemize}
        \item returns $1$ if the edge $\edge$ is $\frac{2\arboricity}{\approxerror}$-heavy.
        \item returns $0$ if the edge $\edge$ is $\frac{\arboricity}{2\approxerror}$-light.
    \end{itemize}
    }

\end{restatable}

\begin{proof}
    For (i), we only consider edges that have at least $\frac{2\arboricity}{\approxerror}$ triangles. In this case, the random variables $Y_i$ ($Y_i$ as in Algorithm~\ref{Algorithm: Heavy Oracle}; $Y_i=1$ if $\vertex_i$ and $\edge$ form a triangle; $0$, otherwise) are i.i.d. Bernoulli random variables taking value $1$ with probability at least $\frac{2\arboricity}{\approxerror ~ \degree{\edge}}$. Hence, we have the following using linearity of expectation:  
    $\Exp\tbrac{Y} = \Exp\tbrac{\frac{1}{\querycount}\sum_{i \in \tbrac{\querycount}} Y_i} = \Exp\tbrac{Y_i} > \frac{2\arboricity}{\approxerror ~ \degree{\edge}} $
    
    As $Y=\frac{1}{\querycount}\sum_i Y_i$, we can upper bound the probability of the algorithm returning $0$ for the edge $\edge$, by a multiplicative Chernoff bound (Lemma~\ref{Lemma: Multiplicative Chernoff Bound}) as follows:
    \begin{align*}
        \Pr\tbrac{Y \leq \frac{\arboricity}{\approxerror\degree{\edge}}} \remove{\complain{\mbox{(Debarshi: why is this } \frac{\arboricity}{\approxerror} \mbox{ and not } \frac{2 \arboricity}{\approxerror}?)}}
        \leq& \Pr\tbrac{Y \leq \fbrac{1-\frac{1}{2}}\Exp\tbrac{Y}} &\left( \Exp\tbrac{Y} > \frac{2\arboricity}{\approxerror\degree{\edge}} \right)\\
        \leq&\exp{\fbrac{-\frac{\querycount\Exp\tbrac{Y}}{12}}} &\text{(by multiplicative Chernoff bound)}\\
        \leq&\exp{\fbrac{-\frac{\querycount\arboricity}{6\degree{\edge}\approxerror}}} &\left( \Exp\tbrac{Y} > \frac{2\arboricity}{\approxerror\degree{\edge}}\right)\\
        \leq & ~ \confidence &\left( \querycount = \frac{16\approxerror \degree{\edge}}{\arboricity}\log\fbrac{\frac{1}{\confidence}}\right)\\
    \end{align*}
    For (ii), we do not have a lower bound on the probability of success $(Y_i = 1)$ in general and hence we cannot obtain a lower bound on $\Exp\tbrac{Y}$. Now, observe that for the edges that participate in very few triangles, the probability of finding triangles (i.e., getting $Y_i = 1$) is low. However, such light edges can still tolerate a high approximation error to be accepted (as a light edge). To account for the trade-off between the lower bound on the probability of $Y_i = 1$ and the upper bound on the approximation factor, we divide the edges participating in at most $\frac{\arboricity}{2\approxerror}$ triangles into $\tbrac{\ceil{\log(\frac{\arboricity}{\approxerror})}}$ buckets with each bucket being defined as the set of edges $\bucket_k = \sbrac{\edge|\frac{\arboricity}{2^{k+1}\approxerror}\leq\numtriangle_\edge < \frac{\arboricity}{2^k\approxerror}}$. For each of these buckets, observe that when \heavyoracle{} is called for an edge belonging to the bucket, we have $\Pr\tbrac{Y_i = 1} = \frac{\triangle_\edge}{\degree{\edge}} \geq \frac{\arboricity}{2^{k+1}\approxerror\degree{\edge}}$, and hence $\Exp\tbrac{Y} = \Exp\tbrac{Y_i} \geq \frac{\arboricity}{2^{k+1}\approxerror\degree{\edge}}$.

    On the other hand, we have $\Exp\tbrac{Y} < \frac{\arboricity}{2^{k}\approxerror\degree{\edge}}$, and hence, $\fbrac{1 + 2^{k-1}}\Exp\tbrac{Y} < \frac{\arboricity}{\approxerror\degree{\edge}}$. Using these observations, we complete the proof by considering any $\edge$ in these buckets:
    \begin{align*}
        \Pr\tbrac{Y \geq \frac{\arboricity}{\approxerror\degree{\edge}}} 
        \leq&\Pr\tbrac{Y \geq \fbrac{1+2^{k-1}}\Exp\tbrac{Y}} &\left( \fbrac{1 + 2^{k-1}}\Exp\tbrac{Y} < \frac{\arboricity}{\approxerror\degree{\edge}}\right)\\
        \leq& \exp{\fbrac{-\frac{2^{2k-2}\querycount\Exp\tbrac{Y}}{2+2^{k-1}}}} &\text{(by multiplicative Chernoff bound)}\\
        \leq& \exp{\fbrac{-\frac{2^{2k}\querycount\Exp\tbrac{Y}}{2^{k+3}}}} &(k \geq 1)\\
        \leq& \exp{\fbrac{-\frac{\querycount\arboricity}{16\approxerror\degree{\edge}}}}&\left( \Exp\tbrac{Y} > \frac{\arboricity}{2^{k+1}\approxerror\degree{\edge}} \right)\\
        \leq&~ \confidence & \left( \querycount = \frac{16\approxerror \degree{\edge}}{\arboricity}\log\fbrac{\frac{1}{\confidence}} \right)
    \end{align*}
\end{proof}

Lemma~\ref{Lemma: Heavy Oracle Algorithm Correctness} shows that Algorithm~\ref{Algorithm: Heavy Oracle} can decide whether an edge is heavy or not with probability at least $1 - \confidence$ using $\frac{16\approxerror ~ \degree{\edge}}{\arboricity}\log\fbrac{\frac{1}{\confidence}}$ iterations; making 2 \degreeq{}, 1 \neighbourq{} and 1 \edgeexistsq{} queries in each iteration. For an individual edge, $\degree{\edge}$ can be at most $\vertexcount-1$, and hence the subroutine might contribute to additional queries for each edge. However, as part of Algorithm~\ref{Algorithm: Random Edge Arboricity Triangle Counting Oracle Triangle Estimate}, \heavyoracle{} is called on a set of edges drawn uniformly at random. Hence, in the following lemma, we bound the expected number of queries made by \heavyoracle{} when called on edges that were drawn uniformly at random.

\begin{restatable}[Query complexity of \heavyoracle{}]{lemma}{heavyoracleqt}\label{Lemma: Heavy Oracle Query Count}
    $\heavyoracle{\fbrac{\edge,\arboricity,\approxerror,\confidence}}$ makes $\constant\approxerror\log\fbrac{\frac{1}{\confidence}}$ queries in expectation on an edge $\edge \sim \uniform\fbrac{\edgeset}$. 
\end{restatable}

\begin{proof}
   For edge $\edge$, $\heavyoracle{\fbrac{\edge,\arboricity,\approxerror,\confidence}}$ makes $\ceil{\frac{16\approxerror \degree{\edge}}{\arboricity}\log\fbrac{\frac{1}{\confidence}}}$ iterations with each iteration making $4$ queries. Hence, the expected number of queries is:
   \begin{align*}
       \Exp_{\edge \sim \uniform\fbrac{\edgeset}}\tbrac{4\ceil{\frac{16\approxerror \degree{\edge}}{\arboricity}\log\fbrac{\frac{1}{\confidence}}}}
       \leq& \frac{64\approxerror}{\arboricity}\log\fbrac{\frac{1}{\confidence}}\Exp_{\edge \sim \uniform\fbrac{\edgeset}}\tbrac{\degree{\edge}}+ 4\\
       =& \frac{64\approxerror}{\arboricity}\log\fbrac{\frac{1}{\confidence}}\frac{1}{\edgecount}\sum_{\edge \in \edgeset} \degree{\edge}+4\\
       \leq& \frac{64\approxerror}{\arboricity}\log\fbrac{\frac{1}{\confidence}}\frac{1}{\edgecount}2\edgecount\arboricity + 4&\text{(by Lemma~\ref{Lemma: deg(e) sum is m * arboricity})}\\
       \leq& 132\approxerror\log\fbrac{\frac{1}{\confidence}} &\left( \approxerror\log\fbrac{\frac{1}{\confidence}} \geq 1 \right)
   \end{align*}
\end{proof}



\subsection{Finalizing the algorithm}
\label{ssec:final-algo}
In this section, we combine the algorithms discussed above to obtain our final algorithm. Our goal is to call \heavyoracle{} with appropriate parameters instead of \exactheavyoracle{} so that the approximation error due to the heavy triangles not being counted remains small. We also derive the query complexity of the algorithm including the queries made through \heavyoracle{}.

\subsubsection{Algorithm with an estimate \texorpdfstring{$\esttriangle$}{}}

\begin{algorithm}[ht!]
    \caption{Triangle Counting Algorithm - with $\esttriangle$}\label{Algorithm: Random Edge Arboricity Triangle Counting Triangle Estimate}
    \begin{algorithmic}[1]
        \Require \degreeq{}, \neighbourq{}, \edgeexistsq{}, and \randedgeq{} query access to a graph $\graph$. Parameters $\esttriangle, \arboricity$, $\approxerror$, $\edgecount$ 
        \State Call Algorithm~\ref{Algorithm: Random Edge Arboricity Triangle Counting Oracle Triangle Estimate} with parameters $\esttriangle, \arboricity$, $\approxerror/2$, and thresholds $\lowerthreshold = 6$, and $\upperthreshold = 24$. We also implement the oracle $\exactheavyoracle\fbrac{\edge,\arboricity,\approxerror}$ through $\heavyoracle\fbrac{\edge,\arboricity,\approxerror/6,\frac{1}{\edgecount\vertexcount}}$. \Comment{Due to Algorithm~\ref{Algorithm: Random Edge Arboricity Triangle Counting Oracle Triangle Estimate} being called with parameter $\approxerror/2$ and the implementation of \heavyoracle{} as stated, the threshold parameters are evaluated w.r.t. call to \heavyoracle{} with parameter $\lowerthreshold = 6$ and $\upperthreshold = 24$}
    \end{algorithmic}
\end{algorithm}

\begin{restatable}[Performance guarantees of Algorithm~\ref{Algorithm: Random Edge Arboricity Triangle Counting Triangle Estimate}]{theorem}{talgolbworks}\label{Theorem: Triangle Estimate ALgo Works}
   Algorithm~\ref{Algorithm: Random Edge Arboricity Triangle Counting Triangle Estimate} makes at most
$\constant\approxerror^{-3} (\edgecount\arboricity/\esttriangle)\log^2\vertexcount$ queries in expectation and returns $\emptriangle$ such that 
$\Pr\tbrac{\emptriangle \in \tbrac{\fbrac{1-\approxerror}\numtriangle,\fbrac{1+\approxerror}\numtriangle}} \geq 1 - \frac{1}{\constant\log \vertexcount}$.
\end{restatable}

\begin{proof}
    Algorithm~\ref{Algorithm: Random Edge Arboricity Triangle Counting Triangle Estimate} will make the same number of calls to \heavyoracle{} as did Algorithm~\ref{Algorithm: Random Edge Arboricity Triangle Counting Oracle Triangle Estimate} to \exactheavyoracle{}. By Theorem~\ref{Theorem: Oracle Triangle Estimate ALgo Works}, Algorithm~\ref{Algorithm: Random Edge Arboricity Triangle Counting Triangle Estimate} makes $\constant(1+\upperthreshold)\approxerror^{-3}(\edgecount\arboricity/\esttriangle)\log\vertexcount$ queries and $\constant(1+\upperthreshold)\approxerror^{-3}\log(\vertexcount)$$(\edgecount\arboricity/\esttriangle)$ calls to $\heavyoracle{}$. By Lemma~\ref{Lemma: Heavy Oracle Query Count}, each call to $\heavyoracle{}$ requires $\constant\approxerror\log \vertexcount$ queries in expectation. Combining with the fact that we have $\upperthreshold = 24$, the algorithm requires $\fbrac{\constant+\constant\approxerror\log\vertexcount}\fbrac{\approxerror^{-3} (\edgecount\arboricity/\esttriangle) \constant\log \vertexcount} \leq \constant\approxerror^{-3} (\edgecount\arboricity/\esttriangle)\log^2\vertexcount$ queries in expectation.

    By Lemma~\ref{Lemma: Heavy Oracle Algorithm Correctness}, we can use $\heavyoracle{\fbrac{\edge,\arboricity,\approxerror/6,\frac{1}{\edgecount\vertexcount}}}$ as an implementation of the idealized $\exactheavyoracle\fbrac{\edge,\arboricity,\approxerror}$. This ensures that a $\frac{6\arboricity}{\approxerror}$-light edge is accepted and a $\frac{24\arboricity}{\approxerror}$-heavy edge is rejected with probability at least $1 - \frac{1}{\edgecount\vertexcount}$. By union bound, the algorithm accepts all $\frac{6\arboricity}{\approxerror}$-light edges and rejects all $\frac{24\arboricity}{\approxerror}$-heavy edges in $\samplededges$ with probability at least $1-\frac{1}{\vertexcount}$. Hence, by Corollary~\ref{Corollary: Lower Bound on Light Triangles} and Lemma~\ref{lemma: E[Y_I] Weight Func Algo}\remove{\complain{(Debarshi: it should need something more than Corollary~\ref{Corollary: Lower Bound on Light Triangles}?)}}, we have, 
    \begin{align}
    \Exp\tbrac{\emptriangle} = \lighttriangles{\frac{6\arboricity}{\approxerror}} \in \tbrac{\fbrac{1-\approxerror/2}\numtriangle,\numtriangle} \label{Eq:Final Light Triangle Bound}
    \end{align}

    By Theorem~\ref{Theorem: Oracle Triangle Estimate ALgo Works}, we have that:
    \begin{align}
        \Pr\tbrac{\emptriangle \in \tbrac{\fbrac{1-\approxerror/2}\lighttriangles{\frac{6\arboricity}{\approxerror}},\fbrac{1+\approxerror/2}\lighttriangles{\frac{6\arboricity}{\approxerror}}}} \geq 1 - \frac{1}{\constant\log \vertexcount}\label{Eq: Final Estimate Bound}
    \end{align}

    Combining Equations~\eqref{Eq:Final Light Triangle Bound} and~\eqref{Eq: Final Estimate Bound}~and using union bound on the event that all oracle calls were executed correctly, for large enough $\vertexcount$:
    \begin{align*}
        \Pr\tbrac{\emptriangle \in \tbrac{\fbrac{1-\approxerror}\numtriangle,\fbrac{1+\approxerror}\numtriangle}} \geq 1 - \frac{1}{\constant\log \vertexcount}&\qedhere
    \end{align*}
\end{proof}

\subsubsection{Getting rid of the assumption on \texorpdfstring{$\esttriangle$}{}}

In this section, we describe the steps to obtain Theorem~\ref{Theorem: Final Upper Bound} from Theorem~\ref{Theorem: Triangle Estimate ALgo Works}. First, we remove the assumption on the knowledge of $\esttriangle$ for Algorithm~\ref{Algorithm: Random Edge Arboricity Triangle Counting Triangle Estimate}. To achieve that, we search for an appropriate choice of $\esttriangle$ starting from $\searchesttriangle = \vertexcount^3$ and halving it each time we fail to find a $\esttriangle$ satisfying the assumption that $\esttriangle \leq 2\numtriangle$.
We must also bound the probability of $\searchesttriangle$ deviating significantly below $\numtriangle$. To ensure that, for each value of $\searchesttriangle$, we run Algorithm~\ref{Algorithm: Random Edge Arboricity Triangle Counting Triangle Estimate} using values of $\esttriangle$ as $\vertexcount^3,\vertexcount^3/2,...,\searchesttriangle$.

\begin{algorithm}
    \caption{Triangle Counting Algorithm}\label{Algorithm: Final Search}
    \begin{algorithmic}[1]
        \Require \degreeq{}, \neighbourq{}, \edgeexistsq{}, and \randedgeq{} query access to a graph $\graph$. Parameters $\arboricity$, $\approxerror$, $\edgecount$ 
        \For{$\searchesttriangle $ in $\tbrac{\vertexcount^3,\vertexcount^3/2,...,1}$}\label{Line: Bar T Range}
            \For{$\esttriangle$ in $\tbrac{\vertexcount^3,\vertexcount^3/2,...,\searchesttriangle }$} \label{Line: EstTriangle Range}
                \For{$i \in \tbrac{2\log\fbrac{\constant\log(\vertexcount)}}$}
                    \State $X_i \gets$ Solution to Algorithm~\ref{Algorithm: Random Edge Arboricity Triangle Counting Triangle Estimate} with parameters $\esttriangle,\arboricity,\approxerror$
                    \State $X \gets \min_i X_i$
                    \If{$X \geq \esttriangle$}
                        \State \Return $X$
                    \EndIf
                \EndFor
            \EndFor
        \EndFor
    \end{algorithmic}
\end{algorithm}

First we bound the probability that the Algorithm~\ref{Algorithm: Final Search} terminates at a choice of $\esttriangle$ that does not satisfy the assumption that $\esttriangle \leq 2\numtriangle$.

\begin{lemma}[Termination bound of Algorithm~\ref{Algorithm: Final Search}]\label{Lemma: Final Search Wrong Termination Bound}
    When $\esttriangle > 2\numtriangle$, the algorithm terminates with probability at most $\frac{1}{\constant\log^2 \vertexcount}$.
\end{lemma}

\begin{proof}
    By Lemma~\ref{lemma: E[Y_I] Weight Func Algo}, we have from Markov's inequality:
    \begin{align*}
        \Pr\tbrac{\emptriangle > 2\numtriangle} \leq \frac{\Exp\tbrac{\emptriangle}}{2\numtriangle} \leq \frac{1}{2}
    \end{align*}
    Now, we have $\Pr\tbrac{X \geq \esttriangle} \leq \Pr\tbrac{\cap_{i}X_i \geq \esttriangle} \leq \frac{1}{\constant\log^2 \vertexcount}$
\end{proof}

\begin{lemma}[Lower bound of $\searchesttriangle$ for Algorithm~\ref{Algorithm: Final Search}]\label{Lemma: Final Search Low Bar T Bound}
    The algorithm reaches a value of $\searchesttriangle  \leq \numtriangle/2^k$ $\fbrac{k \geq 1}$ with probability at most $\frac{1}{\fbrac{\constant\log \vertexcount}^k}$.
\end{lemma}

\begin{proof}
    For every such value of $\searchesttriangle $, we have a $\esttriangle$ in Line~\ref{Line: EstTriangle Range} of Algorithm~\ref{Algorithm: Final Search} such that $\esttriangle \leq \numtriangle/2$. By Theorem~\ref{Theorem: Triangle Estimate ALgo Works}, the algorithm returns an estimate $X_i \geq \numtriangle/2 \geq \esttriangle$ with probability at least $1 - \frac{1}{\constant\log \vertexcount}$. Taking a union bound over possible values of $i \in \tbrac{2\log \fbrac{\constant\log \vertexcount}}$, we have Algorithm~\ref{Algorithm: Final Search} returning an $X$ such that $X \geq \numtriangle/2 \geq \esttriangle$ with probability at least $1 - \frac{1}{\constant\log \vertexcount}$, and the algorithm terminates.
    
    Hence, to get to a $\searchesttriangle  \leq \frac{\numtriangle}{2^k}$, Algorithm~\ref{Algorithm: Final Search} must have failed to terminate for all previous values of $\searchesttriangle $, which happens with probability at most $\frac{1}{\fbrac{\constant\log \vertexcount}^k}$.
\end{proof}

Finally we combine Lemmas~\ref{Lemma: Final Search Wrong Termination Bound} and~\ref{Lemma: Final Search Low Bar T Bound} to obtain the following theorem quantifying the estimation guarantees and query complexity of Algorithm~\ref{Algorithm: Final Search}. The following theorem formalizes the~\Cref{Theorem: Upper Bound Simplified} of Section~\ref{ssec:ourresults}.

\begin{restatable}[Upper bound - Formal]{theorem}{talgoworksfinal}\label{Theorem: Final Upper Bound}
    There exists an algorithm that makes $\bigot{\frac{\edgecount\arboricity\log\frac{1}{\confidence}}{\approxerror^3\numtriangle}}$ queries in expectation, and returns $X \in \tbrac{\fbrac{1-\approxerror}\numtriangle,\fbrac{1+\approxerror}\numtriangle}$ with probability at least $1 - \confidence$.
\end{restatable}

\begin{proof}
    The algorithm runs through at most $\log^2 \vertexcount$ values of $\esttriangle$ such that $\esttriangle > 2\numtriangle$. By Lemma~\ref{Lemma: Final Search Wrong Termination Bound}, the algorithm terminates in each such case with probability at most $1/\constant\log^2\vertexcount$. For each value of $\searchesttriangle  \leq 2\numtriangle$, by Theorem~\ref{Theorem: Triangle Estimate ALgo Works}, the algorithm returns a wrong value with probability at most $1/\constant\log \vertexcount$. There are at most $3\log\vertexcount$ such values of $\searchesttriangle$. Taking an union bound and fixing the constant $\constant$ appropriately, the algorithm returns a correct output with probability at least $5/6$.

    For each  $\searchesttriangle $ in Line~\ref{Line: Bar T Range} of Algorithm~\ref{Algorithm: Final Search}, by Theorem~\ref{Theorem: Triangle Estimate ALgo Works}, the algorithm makes $\bigot{(\edgecount\arboricity/\approxerror^3\searchesttriangle )}$ queries. Hence, till $\searchesttriangle  \leq \numtriangle/2$, the algorithm makes $\bigot{(\edgecount\arboricity/\approxerror^3\numtriangle)}$ queries. Additionally, by Lemma~\ref{Lemma: Final Search Low Bar T Bound}, the queries beyond $\searchesttriangle  \leq \numtriangle$ can be bounded as:
    \begin{align*}
        \sum_{i \in \log\fbrac{\vertexcount}}\frac{1}{\fbrac{\constant\log \vertexcount}^k}\cdot2^k \cdot\bigot{\frac{\edgecount\arboricity}{\approxerror^3\numtriangle}} &=\bigot{\frac{\edgecount\arboricity}{\approxerror^3\numtriangle}}
    \end{align*}
    By the well-known median trick~\cite{chakrabarti2015data}, the median of $\bigo{\log\frac{1}{\confidence}}$ independent runs of Algorithm~\ref{Algorithm: Final Search} establishes the result.
\end{proof}

\section{Lower bound}
\label{sec:lower-bound}
In this section, we prove a lower bound that accounts for $\approxerror$, $\confidence$, and the arboricity $\arboricity$ of the graph by extending the ideas proposed in~\cite{DBLP:conf/approx/AssadiN22}. Our lower bound 
almost matches our proposed upper bound stated in Theorem~\ref{Theorem: Final Upper Bound}. We first state the problem {\sf Popcount Thresholding Problem} (referred to as \ptp{}) in the query framework:
\begin{definition}[$\defptp$]
    Given a string $\alicestring \in \sbrac{0,1}^\stringlength$, the problem $\defptp$ is to  distinguish whether $\alicestring$ is  generated by i.i.d. samples from $\ptpdone$ or $\ptpdtwo$ defined as follows:
    \begin{itemize}
        \item \textbf{$\ptpdone$ :} For all $i \in \tbrac{\stringlength}$, $\alicestring_i$ is set to $1$ with probability $\fbrac{1-2\ptpsep}\frac{\ptpprob}{\stringlength}$, and set to $0$, otherwise.
        \item \textbf{$\ptpdtwo$ :} For all $i \in \tbrac{\stringlength}$, $\alicestring_i$ is set to $1$ with probability $\fbrac{1+2\ptpsep}\frac{\ptpprob}{\stringlength}$, and set to $0$, otherwise.
    \end{itemize}
    The problem is to decide if $\alicestring$ is generated from $\ptpdone$ or $\ptpdtwo$ by querying $\alicestring$ at any of its $\stringlength$ bits.
\end{definition}
We state the following lemma~\cite{DBLP:conf/approx/AssadiN22} quantifying the query complexity of $\defptp$. 
\begin{lemma}[\ptp{} Query Lower Bound~\cite{DBLP:conf/approx/AssadiN22}]\label{Lemma: PTP Query Lower Bound}
    For any $\ptpsep \in (0,1/4)$, $\confidence \in (0,1/100)$, and integers $\stringlength \geq 1$, $\log(1/\confidence)\cdot12/\ptpsep^2\leq \ptpprob \leq \stringlength/6$, $\rqcomplexity{\confidence}{\defptp} \geq \frac{\stringlength\log\fbrac{1/4\confidence}}{24\ptpsep^2\ptpprob}$
  \remove{  \begin{align*}
        \rqcomplexity{\confidence}{\defptp} \geq \frac{\stringlength\log\fbrac{1/4\confidence}}{24\ptpsep^2\ptpprob}
    \end{align*}
    }
    where $\rqcomplexity{\confidence}{\defptp}$ is the randomized query complexity to decide $\defptp{}$ problem with probability at least $1 - \confidence$.
\end{lemma}

We also state the following lemma, establishing bounds on the $\ell_1$-norm of $\alicestring$, which was partially proved in \cite{DBLP:conf/approx/AssadiN22}. 

\begin{restatable}{lemma}{numones}\label{Lemma: PTP Deviation Bound}
     In $\defptp$, for any $\ptpsep \in (0,1/4)$, $\confidence \in (0,1/100)$, and integers $\stringlength \geq 1$, $\log(1/\confidence)\cdot12/\ptpsep^2\leq \ptpprob \leq \stringlength/6$, let $\norm{\alicestring}_1$ denote the number of $1$'s in the string $\alicestring$, then:
   \begin{enumerate}
    \item[(i)] $\Pr\left[\|\alicestring\|_1 > (1 - \ptpsep) \cdot \ptpprob \,\middle|\, \alicestring \sim \ptpdone\right] \leq \confidence$;
    \item[(ii)] $\Pr\left[\|\alicestring\|_1 < (1 + \ptpsep) \cdot \ptpprob \,\middle|\, \alicestring \sim \ptpdtwo\right] \leq \confidence$;
    \item[(iii)] $\Pr\left[\|\alicestring\|_1 < (1 - 4\ptpsep) \cdot \ptpprob \right] \leq \confidence$; \end{enumerate}   
\end{restatable}

In \cite{DBLP:conf/approx/AssadiN22}, (i) and (ii) have been shown using Chernoff bound. (iii) can also be shown by using Chernoff bound. We combine two results to obtain the result. First we state the lemma due to~\cite{DBLP:conf/approx/AssadiN22}:

\begin{lemma}[\cite{DBLP:conf/approx/AssadiN22}]\label{Lemma: PTP Deviation Conditional Bound}
    In $\defptp$, for any $\ptpsep \in (0,1/4)$, $\confidence \in (0,1/100)$, and integers $\stringlength \geq 1$, $\log(1/\confidence)\cdot12/\ptpsep^2\leq \ptpprob \leq \stringlength/6$,
    \begin{enumerate}
        \item[(i)] $\Pr\tbrac{\norm{\alicestring}_1 > \fbrac{1-\ptpsep} \cdot \ptpprob~|~\alicestring \sim \ptpdone} \leq \confidence$;
        \item[(ii)] $\Pr\tbrac{\norm{\alicestring}_1 < \fbrac{1+\ptpsep} \cdot \ptpprob~|~\alicestring \sim \ptpdtwo} \leq \confidence$
    \end{enumerate}
    where $\norm{\alicestring}_1$ denotes the number of $1$'s in the string $\alicestring$.
\end{lemma}

We introduce another lemma to lower bound the $\ell_1$ norm of $\alicestring$ for $\ptpdone$.

\begin{lemma}\label{Lemma: PTP Deviation Lower Bound}
    In $\defptp$, for any $\ptpsep \in (0,1/4)$, $\confidence \in (0,1/100)$, and integers $\stringlength \geq 1$, $\log(1/\confidence)\cdot12/\ptpsep^2\leq \ptpprob \leq \stringlength/6$, $\Pr\tbrac{\norm{\alicestring}_1 < \fbrac{1-4\ptpsep} \cdot \ptpprob~|~\alicestring \sim \ptpdone} \leq \confidence$
    \remove{
    \begin{align*}
        \Pr\tbrac{\norm{\alicestring}_1 < \fbrac{1-4\ptpsep} \cdot \ptpprob~|~\alicestring \sim \ptpdone} \leq \confidence
    \end{align*}
    }
\end{lemma}

\begin{proof}
    By definition of $\ptpdone$ and linearity of expectations, we have:
    \begin{align*}
        \Exp{\tbrac{\norm{\alicestring}_1}} = \sum_{i \in [\stringlength]} \Exp\tbrac{\alicestring_i} = \fbrac{1-2\ptpsep}\ptpprob
    \end{align*}
    Now, given the independence of each $\alicestring_i$, we can use Chernoff bound (Lemma~\ref{Lemma: Multiplicative Chernoff Bound}) to obtain:
    \begin{align*}
        & \Pr\tbrac{\norm{\alicestring}_1 < \fbrac{1-4\ptpsep}\ptpprob} \\
        =&\Pr\tbrac{\norm{\alicestring}_1 - \Exp{\tbrac{\norm{\alicestring}_1}} > \frac{2\ptpsep}{1-2\ptpsep}\Exp{\tbrac{\norm{\alicestring}_1}}}\\
        \leq& \exp{\fbrac{-\frac{4\ptpsep^2\Exp{\tbrac{\norm{\alicestring}_1}}}{3\fbrac{1-2\ptpsep}^2}}}  & \text{(by Chernoff bound)}\\
        \leq& \exp{\fbrac{-\frac{48\ptpsep^2\log\fbrac{1/\confidence}}{3\ptpsep^2\fbrac{1-2\ptpsep}}}} &(\Exp{\tbrac{\norm{\alicestring}_1}} = \fbrac{1-2\ptpsep}\ptpprob, \ptpprob \geq \log(1/\confidence)\cdot12/\ptpsep^2)\\
        \leq& \confidence &(\ptpsep < 1/4)
    \end{align*}
\end{proof}

Combining Lemma~\ref{Lemma: PTP Deviation Conditional Bound} and~\ref{Lemma: PTP Deviation Lower Bound}, we obtain the desired result.

  

\begin{figure}[ht!]
\centering
    \begin{tikzpicture}
        \node[shape=circle,draw=black,minimum height=2.5em] (S) at (0,0) {$\trianglemakerset, \lbarboricity$};
        \node[shape=circle,draw=black,minimum height=2.5em] (A) at (3,2.5) {$\vertexsetA, \frac{\stringlength}{\lbarboricity}$};
        \node[shape=circle,draw=black,minimum height=2.5em] (B) at (-3,2.5) {$\vertexsetB, \frac{\stringlength}{\lbarboricity}$};
        \node[shape=circle,draw=black,minimum height=2.5em] (A') at (3,5) {$\vertexsetAhat, \frac{\stringlength}{\lbarboricity}$};
        \node[shape=circle,draw=black,minimum height=2.5em] (B') at (-3,5) {$\vertexsetBhat, \frac{\stringlength}{\lbarboricity}$};
    
        \path [-] (S) edge node[midway,left] {$(\vertexB,\trianglemaker)$, $\forall \vertexB \in \vertexsetB,\trianglemaker \in \trianglemakerset$} (B);
        \path [-](S) edge node[midway,right] {$(\vertexA,\trianglemaker)$, $\forall \vertexA \in \vertexsetA,\trianglemaker \in \trianglemakerset$} (A);
        \path [dashed](A) edge node[midway,right] {$\fbrac{\vertexA_i,\vertexAhat_{i+j}}$, if $\alicestring_{i,j} = 0$}(A');
        \path [dashed](B) edge node[midway,left] {$\fbrac{\vertexB_i,\vertexBhat_{i+j}}$, if $\alicestring_{i,j} = 0$} (B');
        \path [dashed](B) edge node[midway,above] {$\fbrac{\vertexA_i,\vertexB_{i+j}}$, if $\alicestring_{i,j} = 1$} (A);
        \path [dashed](B') edge node[midway,below] {$\fbrac{\vertexAhat_i,\vertexBhat_{i+j}}$, if $\alicestring_{i,j} = 1$} (A');
    \end{tikzpicture}

    \caption{Lower Bound Graph Construction}\label{Figure: Lower Bound Graph for PTP}
\end{figure}


The following theorem formalizes the~\Cref{Theorem: Lower Bound Simplified} of~\Cref{ssec:ourresults}.

\begin{restatable}[Lower bound - Formal]{theorem}{lbformal}\label{Theorem: Lower Bound on Triangle Counting through PTP}
Any algorithm that solves triangle estimation problem with $\approxerror \leq \frac{1}{8}$ using \degreeq{}, \neighbourq{}, \edgeexistsq{} and \randedgeq{} for any graph with $\edgecount$ edges, $\numtriangle$ triangles, and arboricity $\arboricity$ such that $\numtriangle = \Omega\fbrac{\frac{\arboricity\log\fbrac{1/\confidence}}{\approxerror^2}}$ requires $\Omega\fbrac{\frac{\edgecount\arboricity\log\fbrac{1/\confidence}}{\approxerror^2\numtriangle}}$ queries.    

\end{restatable}

\begin{proof}

We want to show that given a graph $\graph$, it takes $\Omega\fbrac{\frac{\edgecount\arboricity\log\fbrac{1/\confidence}}{\approxerror^2\numtriangle}}$ queries to obtain an estimate $\esttriangle$ such that $\esttriangle \in (1\pm \approxerror)\numtriangle$ with probability at least $1-\confidence$. For a contradiction, let us assume that there exists an algorithm $\trianglealgo$ that, for some arboricity $\lbarboricity$, computes an estimate $\esttriangle$ such that $\esttriangle \in (1\pm \approxerror)\numtriangle$ with probability at least $1-\confidence$ using $\smallo{\frac{\edgecount\lbarboricity\log\fbrac{1/\confidence}}{\approxerror^2\numtriangle}}$ queries.

    To show the contradiction, we design an algorithm that solves $\defptp$ in less than $\frac{\stringlength\log\fbrac{1/4\confidence}}{24\ptpsep^2\ptpprob}$ queries using $\trianglealgo$. Given an instance of $\alicestring$ generated from $\defptp$, we consider a graph $\rqgraph = (\vertexset_x, \edgeset_x)$ defined as: 
    \paragraph*{The vertex set $\vertexset_x$:} The vertex set $\vertexset_x$ consists of $5$ sets of disjoint and independent set of vertices $\vertexsetA, \vertexsetAhat, \vertexsetB,$ $\vertexsetBhat$, and $\trianglemakerset$, with $\size{\vertexsetA} = \size{\vertexsetAhat} = \size{\vertexsetB} = \size{\vertexsetBhat} = \stringlength/\lbarboricity$, and $\size{\trianglemakerset} = \lbarboricity$. Vertices of $\vertexsetA$, $\vertexsetAhat$, $\vertexsetB$ and $\vertexsetBhat$ will be named as $\vertexA$, $\vertexAhat$, $\vertexB$ and $\vertexBhat$, respectively. 
    \paragraph*{The edge set $\edgeset_x$:} There exists an edge between every vertex in $\vertexsetA \cup \vertexsetB$ to every vertex in $\trianglemakerset$. Observe that this makes sure that the arboricity of $\rqgraph$ is at least the average degree of the subgraph induced by $\vertexsetA\cup\vertexsetB\cup\trianglemakerset$, i.e., $\Omega\fbrac{\lbarboricity}$. If a bit of $\alicestring$ is $1$ (resp. $0$), there will be an edge from a vertex in $\vertexsetA$ (resp. $\vertexsetA$) to a vertex in $\vertexsetB$ (resp. $\vertexsetAhat$), and an edge from a vertex in $\vertexsetAhat$ (resp. $\vertexsetB$) to a vertex in $\vertexsetBhat$ (resp. $\vertexsetBhat$).

        We index the problem instance $\alicestring$ of length $\stringlength$ as $\tbrac{\frac{\stringlength}{\lbarboricity}}\times\tbrac{\lbarboricity}$, denoting $\alicestring_{\fbrac{i-1}\times\lbarboricity+j}$ as $\alicestring_{i,j}$, where $i \in \tbrac{\frac{\stringlength}{\lbarboricity}}$ and $j \in \tbrac{\lbarboricity}$. We add an edge $\fbrac{\vertexA_i,\vertexAhat_{i+j}}$, and an edge $\fbrac{\vertexB_i,\vertexBhat_{i+j}}$ if $\alicestring_{i,j} = 0$. We add an edge $\fbrac{\vertexA_i,\vertexB_{i+j}}$, and an edge $\fbrac{\vertexAhat_i,\vertexBhat_{i+j}}$ if $\alicestring_{i,j} = 1$. All $(i+j)$ additions here are modulo $\tbrac{\frac{\stringlength}{\lbarboricity}}$.

    {Consider the orientation of edges such that all edges from $\trianglemakerset$ is oriented outwards, and other edges are oriented arbitrarily. This gives a DAG with the out-degree of each edge being bounded by $2\lbarboricity$. As a DAG with out-degree at most $\arboricity$ ensures that the arboricity of the graph $G_\alicestring$ is at most $\arboricity$\cite{DBLP:conf/soda/EdenRS20,MatulaBeck/JACM/83/SmallestLastOrderingArboricity,CHROBAK1991243}, the graph has arboricity at most $2\lbarboricity$.} Hence, the arboricity of $G_\alicestring$ is $\Theta(\alpha)$, as we have already shown it to be $\Omega(\alpha)$.

{From the construction, observe that the number of edges between $A$ and $B$ is $\norm{\alicestring}_1$. Moreover, each edge between  $\vertexsetA$ and $\vertexsetB$ adds $\lbarboricity$ triangles.} Hence, the number of triangles in $\rqgraph$ is exactly $\norm{\alicestring}_1\lbarboricity$. Also note that we have added $2\stringlength$ edges between $\vertexsetA\cup\vertexsetB$ and $\trianglemakerset$, and further $2\stringlength$ edges {(in the subgraph induced by $V_x \setminus S$)} according to the entries of $\alicestring$, 2 for each bit $\alicestring_i$. Hence, we have $\lbedgecount = 4\stringlength$, and we fix $\approxerror = \ptpsep$.

We now describe $\trianglealgoptp$ an algorithm that uses $\trianglealgo$ to solve $\defptp$ using less than $\frac{\stringlength\log\fbrac{1/4\confidence}}{24\ptpsep^2\ptpprob}$ queries. Given a string $\alicestring$, we generate $\rqgraph$ and estimate the number of triangles $\emptriangle$ by calling $\trianglealgo$ using $\frac{\edgecount\log\fbrac{1/4\confidence}}{200\approxerror^2\ptpprob} = \frac{\stringlength\log(1/4\confidence)}{50\ptpsep^2\ptpprob}$ queries. If $\emptriangle < \fbrac{1-\ptpsep^2}\ptpprob\lbarboricity$, it outputs $\alicestring \sim \ptpdone$; otherwise, it outputs $\alicestring \sim \ptpdtwo$. Also, observe that by Lemma~\ref{Lemma: PTP Deviation Bound}, we have  $\numtriangle = \lbarboricity\norm{\alicestring}_1 \geq \lbarboricity(1-4\ptpsep)\ptpprob \geq \frac{\lbarboricity\ptpprob}{2}$ with probability at least $1-\confidence$, where the last inequality follows from the fact that $\ptpsep = \approxerror \leq \frac{1}{4}$. {Lemma~\ref{Lemma: PTP Deviation Bound} is applicable due to the fact that $\numtriangle = \Omega\fbrac{ \frac{\lbarboricity\log\fbrac{1/\confidence}}{\approxerror^2}}$.} Hence, we have $\frac{\edgecount\log\fbrac{1/4\confidence}}{200\approxerror^2\ptpprob} \geq \frac{\edgecount\lbarboricity\log\fbrac{1/4\confidence}}{200\approxerror^2\numtriangle}$. Thus, $\trianglealgo$ is allowed to make $\smallo{\frac{\edgecount\lbarboricity\log\fbrac{1/\confidence}}{\approxerror^2\numtriangle}}$ queries under the given query bound.

By Lemma~\ref{Lemma: PTP Deviation Bound}, to distinguish between whether $\alicestring$ is generated from $\ptpdone$ or $\ptpdtwo$ with probability at least $1 - \confidence$, it suffices to show that we can use $\trianglealgo$ to distinguish between the instance where $\norm{\alicestring}_1 > \fbrac{1+\ptpsep}\ptpprob$ and the instance where $\norm{\alicestring}_1 < \fbrac{1-\ptpsep}\ptpprob$ instance with probability at least $1-\frac{\confidence}{2}$. We now show that $\trianglealgoptp$ outputs $\ptpdtwo$ with probability at least $1-\frac{\confidence}{2}$ given $\norm{\alicestring}_1 > \fbrac{1+\ptpsep}\ptpprob$, and outputs $\ptpdone$ with probability at least $1-\frac{\confidence}{2}$ given $\norm{\alicestring}_1 < \fbrac{1-\ptpsep}\ptpprob$. We consider the two cases separately: 

    {\bf Case I }$\left(\norm{\alicestring}_1 > \fbrac{1+\ptpsep}\ptpprob\right)$: Our construction ensures that the number of triangles in $\rqgraph$ is $\numtriangle = \norm{\alicestring}_1\lbarboricity > \fbrac{1+\ptpsep}\ptpprob\lbarboricity = \fbrac{1+\approxerror}\ptpprob\lbarboricity$. Additionally, by our assumption on $\trianglealgo$, it outputs an estimate $\emptriangle \geq (1-\approxerror)\numtriangle$ with probability $1 - \frac{\confidence}{2}$. Thus, we have $\esttriangle \geq \fbrac{1 - \approxerror^2}\ptpprob\lbarboricity$ with probability at least $1-\frac{\confidence}{2}$.

    {\bf Case II }$\left(\norm{\alicestring}_1 < \fbrac{1-\ptpsep}\ptpprob\right)$: Our construction ensures that the number of triangles in $\rqgraph$ is $\numtriangle = \norm{\alicestring}_1\lbarboricity < \fbrac{1-\ptpsep}\ptpprob\lbarboricity = \fbrac{1-\approxerror}\ptpprob\lbarboricity$. Additionally, by our assumption on $\trianglealgo$, it outputs an estimate $\emptriangle \leq (1+\approxerror)\numtriangle$ with probability $1 - \frac{\confidence}{2}$. Thus, we have $\esttriangle \leq \fbrac{1 - \approxerror^2}\ptpprob\lbarboricity$ with probability at least $1-\frac{\confidence}{2}$.

   Now we state how to simulate the required queries used by the algorithm:
    \begin{itemize}
        \item \textbf{$\degreeq{\fbrac{\vertex}}$:} For a vertex $\vertex \in \vertexsetA\cup\vertexsetB\cup\vertexsetAhat\cup\vertexsetBhat$, return $\lbarboricity$, and for a vertex $\vertex \in \trianglemakerset$, return $\frac{2\stringlength}{\lbarboricity}$. Hence, for \degreeq{} queries, no queries are made to the string.
        

        \item \textbf{$\neighbourq{\fbrac{\vertex,i}}$:} For a vertex $\vertex \in \vertexsetA\cup\vertexsetB\cup\vertexsetAhat\cup\vertexsetBhat$, w.l.o.g assume $\vertex$ to be $\vertexA_j \in \vertexsetA$. Return $\vertexAhat_{i+j}$ if $\alicestring_{i,j} = 0$, and $\vertexB_{i+j}$ otherwise. For a vertex $\vertex \in \trianglemakerset$, if $i \leq \frac{\edgecount}{\lbarboricity}$, return $\vertexA_{i}$, else return $\vertexB_{i - \frac{\edgecount}{\lbarboricity}}$. Hence, for \neighbourq{} queries, a single query is made to the string.
        
        \item \textbf{\edgeexistsq$\fbrac{\altvertex,\vertex}$:} Given a vertex pair $\fbrac{\vertex,\altvertex}$, there can be 3 cases: (a) If the vertex $\vertex \in \vertexsetA\cup\vertexsetB$, w.l.o.g assume $\vertex$ to be $\vertexA_i \in \vertexsetA$, if $\altvertex \in \trianglemakerset$, return $1$, else if $\altvertex = \vertexB_j \in \vertexsetB$ and $\alicestring_{i,j-i} = 1$, return $1$, else if $\altvertex = \vertexAhat_j \in \vertexsetAhat$ and $\alicestring_{i,j-i} = 0$, return $1$, else return $0$. (b) If a vertex $\vertex \in \vertexsetAhat\cup\vertexsetBhat$, w.l.o.g assume the $\vertex$ to be $\vertexAhat_i \in \vertexsetAhat$, if $\altvertex = \vertexA_j \in \vertexsetA$, return $1$ if $\alicestring_{j,i-j} = 0$, else if $\altvertex = \vertexBhat_j \in \vertexsetBhat$, return $1$ if $\alicestring_{j,i-j} = 1$, else return $0$. If $\vertex \in \trianglemakerset$, return $1$ if $\altvertex \in \vertexsetA\cup\vertexsetB$, else return $0$. (c) If a vertex $\vertex \in \trianglemakerset$, return $1$ if $\altvertex \in \vertexsetA \cup \vertexsetB$, return $0$ otherwise. Hence, for \edgeexistsq{} queries, a single query is made to the string.

        \item \textbf{\randedgeq{}:} Sample a vertex $\vertex$ with probability $\frac{\degree{\vertex}}{4\stringlength}$. Choose $i \in \degree{\fbrac{\vertex}}$ uniformly at random, query $\altvertex = \neighbourq{\fbrac{\vertex,i}}$ and return $\fbrac{\vertex,\altvertex}$. Hence, for \randedgeq{} queries, a single query is made to the string for the $\neighbourq{}$ query. \qedhere
    \end{itemize}

\end{proof}

The following corollary is a direct implication from Theorem~\ref{Theorem: Lower Bound on Triangle Counting through PTP} due to the fact that we consider a weaker model (without access to \randedgeq{} queries) compared to Theorem~\ref{Theorem: Lower Bound on Triangle Counting through PTP}.

\begin{corollary}[Lower Bound for Local Queries]\label{Theorem: Lower Bound on Triangle Counting w/o Random Edge}
Any algorithm that solves triangle estimation problem with $\approxerror \leq \frac{1}{4}$ using \degreeq{}, \neighbourq{}, and \edgeexistsq{} requires $\Omega\fbrac{\frac{\edgecount\arboricity\log\fbrac{1/\confidence}}{\approxerror^2\numtriangle}}$ queries.    
\end{corollary}

\section{Conclusion}

In this paper, we present an almost optimal algorithm for triangle counting in graphs of bounded arboricity, assuming access to both local queries and \randedgeq{} queries. In particular, the query complexity of our algorithm is $O(m\alpha/T)$.%
\footnote{In this section, we suppress the dependence on polylogarithmic factors, for simplicity of presentation.}
Extending our result to clique counting remains an interesting open question. Eden et al.~\cite{DBLP:conf/soda/EdenRS20} showed that the number of $k$-cliques can be estimated using $O\left(n/T_k^{1/k}+m^{k/2}/T_k\right)$ local queries, where $T_k$ denotes the number of $k$-cliques in the graph. When \randedgeq{} queries are also available, Assadi et al.~\cite{assadi2018simple} established a bound of $O\left(m^{k/2}/T_k \right)$. For graphs with arboricity bounded by $\alpha$, Eden et al.~\cite{DBLP:conf/soda/EdenRS20} showed that $O\left(n/T_k^{1/k}+m\alpha^{k-2}/T_k\right)$ local queries suffice. We believe that, when \randedgeq{} queries are available in addition to local queries, the bound can be improved to $O\left(m\alpha^{k-2}/T_k\right)$.  However, establishing the complexity of clique counting (when \randedgeq{} is available along with local queries) remains an open question. A similar open question was posed in the streaming model by Bera and Seshadri~\cite{BeraSeshadriStreamingDegeneracy}.
\color{black}


\bibliography{ref}

@InProceedings{DBLP:conf/approx/EdenR18,
  author =	{Eden, Talya and Rosenbaum, Will},
  title =	{{Lower Bounds for Approximating Graph Parameters via Communication Complexity}},
  booktitle =	{Approximation, Randomization, and Combinatorial Optimization. Algorithms and Techniques (APPROX/RANDOM 2018)},
  pages =	{11:1--11:18},
  series =	{Leibniz International Proceedings in Informatics (LIPIcs)},
  ISBN =	{978-3-95977-085-9},
  ISSN =	{1868-8969},
  year =	{2018},
  volume =	{116},
  editor =	{Blais, Eric and Jansen, Klaus and D. P. Rolim, Jos\'{e} and Steurer, David},
  publisher =	{Schloss Dagstuhl -- Leibniz-Zentrum f{\"u}r Informatik},
  address =	{Dagstuhl, Germany},
  URL =		{https://drops.dagstuhl.de/entities/document/10.4230/LIPIcs.APPROX-RANDOM.2018.11},
  URN =		{urn:nbn:de:0030-drops-94156},
  doi =		{10.4230/LIPIcs.APPROX-RANDOM.2018.11}
}

@InProceedings{jayaram_et_al:LIPIcs.APPROX/RANDOM.2021.11,
  author =	{Jayaram, Rajesh and Kallaugher, John},
  title =	{{An Optimal Algorithm for Triangle Counting in the Stream}},
  booktitle =	{{APPROX/RANDOM} 2021},
  pages =	{11:1--11:11},
  series =	{Leibniz International Proceedings in Informatics (LIPIcs)},
  ISBN =	{978-3-95977-207-5},
  ISSN =	{1868-8969},
  year =	{2021},
  volume =	{207},
  publisher =	{Schloss Dagstuhl -- Leibniz-Zentrum f{\"u}r Informatik},
  address =	{Dagstuhl, Germany},
  URL =		{https://drops-dev.dagstuhl.de/entities/document/10.4230/LIPIcs.APPROX/RANDOM.2021.11},
  URN =		{urn:nbn:de:0030-drops-147046},
  doi =		{10.4230/LIPIcs.APPROX/RANDOM.2021.11}
}

@article{Dana_Ron_Triangle_Counting,
author = {Eden, Talya and Levi, Amit and Ron, Dana and Seshadhri, C.},
title = {Approximately Counting Triangles in Sublinear Time},
journal = {SIAM Journal on Computing},
volume = {46},
number = {5},
pages = {1603-1646},
year = {2017},
doi = {10.1137/15M1054389},
URL = {https://doi.org/10.1137/15M1054389},
eprint = {https://doi.org/10.1137/15M1054389}
}

@InProceedings{tetek:LIPIcs.ICALP.2022.107,
  author =	{T\v{e}tek, Jakub},
  title =	{{Approximate Triangle Counting via Sampling and Fast Matrix Multiplication}},
  booktitle =	{49th International Colloquium on Automata, Languages, and Programming (ICALP 2022)},
  pages =	{107:1--107:20},
  series =	{Leibniz International Proceedings in Informatics (LIPIcs)},
  ISBN =	{978-3-95977-235-8},
  ISSN =	{1868-8969},
  year =	{2022},
  volume =	{229},
  editor =	{Boja\'{n}czyk, Miko{\l}aj and Merelli, Emanuela and Woodruff, David P.},
  publisher =	{Schloss Dagstuhl -- Leibniz-Zentrum f{\"u}r Informatik},
  address =	{Dagstuhl, Germany},
  URL =		{https://drops.dagstuhl.de/entities/document/10.4230/LIPIcs.ICALP.2022.107},
  URN =		{urn:nbn:de:0030-drops-164485},
  doi =		{10.4230/LIPIcs.ICALP.2022.107}
}

@article{GRS11,
author = {Gonen, Mira and Ron, Dana and Shavitt, Yuval},
title = {Counting Stars and Other Small Subgraphs in Sublinear-Time},
journal = {SIAM Journal on Discrete Mathematics},
volume = {25},
number = {3},
pages = {1365-1411},
year = {2011},
doi = {10.1137/100783066},
URL = {https://doi.org/10.1137/100783066},
eprint = {https://doi.org/10.1137/100783066}
}

@inproceedings{BOV_13_Streaming_triangle_counting_hardness,
author = {Braverman, Vladimir and Ostrovsky, Rafail and Vilenchik, Dan},
title = {How hard is counting triangles in the streaming model?},
booktitle =	{40th International Colloquium on Automata, Languages, and Programming (ICALP 2013)},
year = {2013},
isbn = {9783642392054},
publisher = {Springer-Verlag},
address = {Berlin, Heidelberg},
url = {https://doi.org/10.1007/978-3-642-39206-1_21},
doi = {10.1007/978-3-642-39206-1_21},
pages = {244–254},
numpages = {11},
location = {Riga, Latvia},
series = {ICALP'13}
}

@inproceedings{KallaugherP17,
author = {Kallaugher, John and Price, Eric},
title = {A hybrid sampling scheme for triangle counting},
year = {2017},
publisher = {Society for Industrial and Applied Mathematics},
address = {USA},
booktitle = {Proceedings of the Twenty-Eighth Annual ACM-SIAM Symposium on Discrete Algorithms},
pages = {1778–1797},
numpages = {20},
location = {Barcelona, Spain},
series = {SODA '17}
}

@inproceedings{assadi2018simple,
    author =	{Assadi, Sepehr and Kapralov, Michael and Khanna, Sanjeev},
  title =	{{A Simple Sublinear-Time Algorithm for Counting Arbitrary Subgraphs via Edge Sampling}},
  booktitle =	{10th Innovations in Theoretical Computer Science Conference (ITCS 2019)},
  pages =	{6:1--6:20},
  series =	{Leibniz International Proceedings in Informatics (LIPIcs)},
  ISBN =	{978-3-95977-095-8},
  ISSN =	{1868-8969},
  year =	{2019},
  volume =	{124},
  editor =	{Blum, Avrim},
  publisher =	{Schloss Dagstuhl -- Leibniz-Zentrum f{\"u}r Informatik},
  address =	{Dagstuhl, Germany},
  URL =		{https://drops.dagstuhl.de/entities/document/10.4230/LIPIcs.ITCS.2019.6},
  URN =		{urn:nbn:de:0030-drops-100996},
  doi =		{10.4230/LIPIcs.ITCS.2019.6}
}

@inproceedings{DBLP:conf/soda/EdenRS20,
author = {Eden, Talya and Ron, Dana and Seshadhri, C.},
title = {Faster sublinear approximation of the number of k-cliques in low-arboricity graphs},
year = {2020},
publisher = {Society for Industrial and Applied Mathematics},
address = {USA},
booktitle = {Proceedings of the Thirty-First Annual ACM-SIAM Symposium on Discrete Algorithms},
pages = {1467–1478},
numpages = {12},
location = {Salt Lake City, Utah},
series = {SODA '20}
}

@article{DBLP:journals/siamcomp/ChibaN85,
  author       = {Norishige Chiba and
                  Takao Nishizeki},
  title        = {Arboricity and Subgraph Listing Algorithms},
  journal      = {{SIAM} J. Comput.},
  volume       = {14},
  number       = {1},
  pages        = {210--223},
  year         = {1985},
  url          = {https://doi.org/10.1137/0214017},
  doi          = {10.1137/0214017},
  bibsource    = {dblp computer science bibliography, https://dblp.org}
}

@book{Mitzenmacher_Upfal_2005,
place={Cambridge},
title={Probability and Computing: Randomized Algorithms and Probabilistic Analysis},
publisher={Cambridge University Press},
address = {Cambridge},
author={Mitzenmacher, Michael and Upfal, Eli}, 
year={2005}}

@inproceedings{DBLP:conf/soda/KopelowitzPP16,
author = {Kopelowitz, Tsvi and Pettie, Seth and Porat, Ely},
title = {Higher lower bounds from the 3SUM conjecture},
year = {2016},
isbn = {9781611974331},
publisher = {Society for Industrial and Applied Mathematics},
address = {USA},
booktitle = {Proceedings of the Twenty-Seventh Annual ACM-SIAM Symposium on Discrete Algorithms},
pages = {1272–1287},
numpages = {16},
location = {Arlington, Virginia},
series = {SODA '16}
}

@inproceedings{DBLP:conf/focs/WilliamsX20,
  author       = {Virginia {Vassilevska Williams} and
                  Yinzhan Xu},
  editor       = {Sandy Irani},
  title        = {Monochromatic Triangles, Triangle Listing and {APSP}},
  booktitle    = {61st {IEEE} Annual Symposium on Foundations of Computer Science, {FOCS}
                  2020, Durham, NC, USA, November 16-19, 2020},
  pages        = {786--797},
  publisher    = {{IEEE}},
  year         = {2020},
  url          = {https://doi.org/10.1109/FOCS46700.2020.00078},
    address = {Durham,NC,USA},
  doi          = {10.1109/FOCS46700.2020.00078},
  bibsource    = {dblp computer science bibliography, https://dblp.org}
}

@article{DBLP:journals/pvldb/PavanTTW13,
  author       = {A. Pavan and
                  Kanat Tangwongsan and
                  Srikanta Tirthapura and
                  Kun{-}Lung Wu},
  title        = {Counting and Sampling Triangles from a Graph Stream},
  journal      = {Proc. {VLDB} Endow.},
  volume       = {6},
  number       = {14},
  pages        = {1870--1881},
  year         = {2013},
  url          = {http://www.vldb.org/pvldb/vol6/p1870-aduri.pdf},
  doi          = {10.14778/2556549.2556569},
  bibsource    = {dblp computer science bibliography, https://dblp.org}
}

@article{DBLP:journals/ipl/PaghT12,
  author       = {Rasmus Pagh and
                  Charalampos E. Tsourakakis},
  title        = {Colorful triangle counting and a MapReduce implementation},
  journal      = {Inf. Process. Lett.},
  volume       = {112},
  number       = {7},
  pages        = {277--281},
  year         = {2012},
  url          = {https://doi.org/10.1016/j.ipl.2011.12.007},
  doi          = {10.1016/J.IPL.2011.12.007},
  bibsource    = {dblp computer science bibliography, https://dblp.org}
}

@inproceedings{DBLP:conf/stacs/BeraC17,
  author       = {Suman K. Bera and
                  Amit Chakrabarti},
  editor       = {Heribert Vollmer and
                  Brigitte Vall{\'{e}}e},
  title        = {Towards Tighter Space Bounds for Counting Triangles and Other Substructures
                  in Graph Streams},
  booktitle    = {34th Symposium on Theoretical Aspects of Computer Science, {STACS}
                  2017, March 8-11, 2017, Hannover, Germany},
    address= {Hannover, Germany},
  series       = {LIPIcs},
  volume       = {66},
  pages        = {11:1--11:14},
  publisher    = {Schloss Dagstuhl - Leibniz-Zentrum f{\"{u}}r Informatik},
  year         = {2017},
  url          = {https://doi.org/10.4230/LIPIcs.STACS.2017.11},
  doi          = {10.4230/LIPICS.STACS.2017.11},
  bibsource    = {dblp computer science bibliography, https://dblp.org}
}

@inproceedings{DBLP:conf/pods/McGregorVV16,
author = {McGregor, Andrew and Vorotnikova, Sofya and Vu, Hoa T.},
title = {Better Algorithms for Counting Triangles in Data Streams},
year = {2016},
isbn = {9781450341912},
publisher = {Association for Computing Machinery},
address = {New York, NY, USA},
url = {https://doi.org/10.1145/2902251.2902283},
doi = {10.1145/2902251.2902283},
booktitle = {Proceedings of the 35th ACM SIGMOD-SIGACT-SIGAI Symposium on Principles of Database Systems},
pages = {401–411},
numpages = {11},
location = {San Francisco, California, USA},
series = {PODS '16}
}

@inproceedings{DBLP:conf/approx/AssadiN22,
  author       = {Sepehr Assadi and
                  Hoai{-}An Nguyen},
  editor       = {Amit Chakrabarti and
                  Chaitanya Swamy},
  title        = {Asymptotically Optimal Bounds for Estimating H-Index in Sublinear
                  Time with Applications to Subgraph Counting},
  booktitle    = {Approximation, Randomization, and Combinatorial Optimization. Algorithms
                  and Techniques, {APPROX/RANDOM} 2022, September 19-21, 2022, University
                  of Illinois, Urbana-Champaign, {USA} (Virtual Conference)},
  address      = {Illinois, Urbana-Champaign, {USA} (Virtual Conference)}, 
  series       = {LIPIcs},
  volume       = {245},
  pages        = {48:1--48:20},
  publisher    = {Schloss Dagstuhl - Leibniz-Zentrum f{\"{u}}r Informatik},
  year         = {2022},
  url          = {https://doi.org/10.4230/LIPIcs.APPROX/RANDOM.2022.48},
  doi          = {10.4230/LIPICS.APPROX/RANDOM.2022.48},
  bibsource    = {dblp computer science bibliography, https://dblp.org}
}

@inproceedings{BeraSeshadriStreamingDegeneracy,
author = {Bera, Suman K. and Seshadhri, C.},
title = {How the Degeneracy Helps for Triangle Counting in Graph Streams},
year = {2020},
isbn = {9781450371087},
publisher = {Association for Computing Machinery},
address = {New York, NY, USA},
url = {https://doi.org/10.1145/3375395.3387665},
doi = {10.1145/3375395.3387665},
booktitle = {Proceedings of the 39th ACM SIGMOD-SIGACT-SIGAI Symposium on Principles of Database Systems},
pages = {457–467},
numpages = {11},
location = {Portland, OR, USA},
series = {PODS'20}
}

@article{Blais2012,
  title = {Property Testing Lower Bounds via Communication Complexity},
  volume = {21},
  ISSN = {1420-8954},
  url = {http://dx.doi.org/10.1007/s00037-012-0040-x},
  DOI = {10.1007/s00037-012-0040-x},
  number = {2},
  journal = {computational complexity},
  publisher = {Springer Science and Business Media LLC},
  author = {Blais,  Eric and Brody,  Joshua and Matulef,  Kevin},
  year = {2012},
  month = may,
  pages = {311–358}
}

@inbook{Goldreich2020,
title = "On the communication complexity methodology for proving lower bounds on the query complexity of property testing",
author = "Oded Goldreich",
year = "2020",
month = apr,
day = "4",
doi = "10.1007/978-3-030-43662-9_7",
language = "English",
series = "Lecture Notes in Computer Science (including subseries Lecture Notes in Artificial Intelligence and Lecture Notes in Bioinformatics)",
pages = "87--118",
booktitle = "Computational Complexity and Property Testing",
publisher = "Springer Verlag",
address = "Germany",
}

@article{DBLP:journals/siamdm/AlonKKR08,
  author       = {Noga Alon and
                  Tali Kaufman and
                  Michael Krivelevich and
                  Dana Ron},
  title        = {Testing Triangle-Freeness in General Graphs},
  journal      = {{SIAM} J. Discret. Math.},
  volume       = {22},
  number       = {2},
  pages        = {786--819},
  year         = {2008},
  url          = {https://doi.org/10.1137/07067917X},
  doi          = {10.1137/07067917X},
  bibsource    = {dblp computer science bibliography, https://dblp.org}
}

@INPROCEEDINGS{AlonTestingFOCS2013,
author = { Alon, N. },
booktitle = {IEEE 54th Annual Symposium on Foundations of Computer Science},
title = {{ Testing Subgraphs in Large Graphs }},
year = {2013},
volume = {},
ISSN = {},
pages = {434},
doi = {10.1109/SFCS.2001.959918},
url = {https://doi.ieeecomputersociety.org/10.1109/SFCS.2001.959918},
publisher = {IEEE Computer Society},
address = {Los Alamitos, CA, USA},
month =Oct}

@InProceedings{SWExperimentalListingTriangles2005,
author="Schank, Thomas
and Wagner, Dorothea",
editor="Nikoletseas, Sotiris E.",
title="Finding, Counting and Listing All Triangles in Large Graphs, an Experimental Study",
booktitle="Experimental and Efficient Algorithms",
year="2005",
publisher="Springer Berlin Heidelberg",
address="Berlin, Heidelberg",
pages="606--609",
isbn="978-3-540-32078-4"
}

@inproceedings{DBLP:conf/soda/Bar-YossefKS02,
  author = {Bar-Yossef, Ziv and Kumar, Ravi and Sivakumar, D.},
title = {Reductions in streaming algorithms, with an application to counting triangles in graphs},
year = {2002},
isbn = {089871513X},
publisher = {Society for Industrial and Applied Mathematics},
address = {USA},
booktitle = {Proceedings of the Thirteenth Annual ACM-SIAM Symposium on Discrete Algorithms},
pages = {623–632},
numpages = {10},
location = {San Francisco, California},
series = {SODA '02}
}

@inproceedings{AhnGMPODS2012,
author = {Ahn, Kook Jin and Guha, Sudipto and McGregor, Andrew},
title = {Graph sketches: sparsification, spanners, and subgraphs},
year = {2012},
isbn = {9781450312486},
publisher = {Association for Computing Machinery},
address = {New York, NY, USA},
url = {https://doi.org/10.1145/2213556.2213560},
doi = {10.1145/2213556.2213560},
booktitle = {Proceedings of the 31st ACM SIGMOD-SIGACT-SIGAI Symposium on Principles of Database Systems},
pages = {5–14},
numpages = {10},
location = {Scottsdale, Arizona, USA},
series = {PODS '12}
}

@article{DBLP:journals/algorithmica/AlonYZ97,
  author       = {Noga Alon and
                  Raphael Yuster and
                  Uri Zwick},
  title        = {Finding and Counting Given Length Cycles},
  journal      = {Algorithmica},
  volume       = {17},
  number       = {3},
  pages        = {209--223},
  year         = {1997},
  url          = {https://doi.org/10.1007/BF02523189},
  doi          = {10.1007/BF02523189},
  bibsource    = {dblp computer science bibliography, https://dblp.org}
}

@article{TriangleCountingDynamicGraphALgorithmica2016,
author = {Bulteau, Laurent and Froese, Vincent and Kutzkov, Konstantin and Pagh, Rasmus},
title = {Triangle Counting in Dynamic Graph Streams},
year = {2016},
issue_date = {September 2016},
publisher = {Springer-Verlag},
address = {Berlin, Heidelberg},
volume = {76},
number = {1},
issn = {0178-4617},
url = {https://doi.org/10.1007/s00453-015-0036-4},
doi = {10.1007/s00453-015-0036-4},
journal = {Algorithmica},
month = sep,
pages = {259–278},
numpages = {20}
}

@inproceedings{BuriolFLPODS06,
author = {Buriol, Luciana S. and Frahling, Gereon and Leonardi, Stefano and Marchetti-Spaccamela, Alberto and Sohler, Christian},
title = {Counting triangles in data streams},
year = {2006},
isbn = {1595933182},
publisher = {Association for Computing Machinery},
address = {New York, NY, USA},
url = {https://doi.org/10.1145/1142351.1142388},
doi = {10.1145/1142351.1142388},
booktitle = {Proceedings of the Twenty-Fifth ACM SIGMOD-SIGACT-SIGART Symposium on Principles of Database Systems},
pages = {253–262},
numpages = {10},
location = {Chicago, IL, USA},
series = {PODS '06}
}

@article{JhaSeshadriStreamingTransitivity,
author = {Jha, Madhav and Seshadhri, C. and Pinar, Ali},
title = {A Space-Efficient Streaming Algorithm for Estimating Transitivity and Triangle Counts Using the Birthday Paradox},
year = {2015},
issue_date = {April 2015},
publisher = {Association for Computing Machinery},
address = {New York, NY, USA},
volume = {9},
number = {3},
issn = {1556-4681},
url = {https://doi.org/10.1145/2700395},
doi = {10.1145/2700395},
journal = {ACM Trans. Knowl. Discov. Data},
month = feb,
articleno = {15},
numpages = {21}
}

@article{ShinDynamicStreamTriangle,
author = {Shin, Kijung and Oh, Sejoon and Kim, Jisu and Hooi, Bryan and Faloutsos, Christos},
title = {Fast, Accurate and Provable Triangle Counting in Fully Dynamic Graph Streams},
year = {2020},
issue_date = {April 2020},
publisher = {Association for Computing Machinery},
address = {New York, NY, USA},
volume = {14},
number = {2},
issn = {1556-4681},
url = {https://doi.org/10.1145/3375392},
doi = {10.1145/3375392},
journal = {ACM Trans. Knowl. Discov. Data},
month = feb,
articleno = {12},
numpages = {39}
}

@article{CORMODE2017SecondLook,
title = {A second look at counting triangles in graph streams (corrected)},
journal = {Theoretical Computer Science},
volume = {683},
pages = {22-30},
year = {2017},
issn = {0304-3975},
doi = {https://doi.org/10.1016/j.tcs.2016.06.020},
url = {https://www.sciencedirect.com/science/article/pii/S0304397516302560},
author = {Graham Cormode and Hossein Jowhari},
}

@inproceedings{KaneMelhornArbitrarySubgrapgICALP12,
author = {Kane, Daniel M. and Mehlhorn, Kurt and Sauerwald, Thomas and Sun, He},
title = {Counting arbitrary subgraphs in data streams},
year = {2012},
isbn = {9783642315848},
publisher = {Springer-Verlag},
address = {Berlin, Heidelberg},
url = {https://doi.org/10.1007/978-3-642-31585-5_53},
doi = {10.1007/978-3-642-31585-5_53},
booktitle = {Proceedings of the 39th International Colloquium Conference on Automata, Languages, and Programming - Volume Part II},
pages = {598–609},
numpages = {12},
location = {Warwick, UK},
series = {ICALP'12}
}

@inproceedings{JainSeshadriTuranWWW17,
author = {Jain, Shweta and Seshadhri, C.},
title = {A Fast and Provable Method for Estimating Clique Counts Using Tur\'{a}n's Theorem},
year = {2017},
isbn = {9781450349130},
publisher = {International World Wide Web Conferences Steering Committee},
address = {Republic and Canton of Geneva, CHE},
url = {https://doi.org/10.1145/3038912.3052636},
doi = {10.1145/3038912.3052636},
booktitle = {Proceedings of the 26th International Conference on World Wide Web},
pages = {441–449},
numpages = {9},
location = {Perth, Australia},
series = {WWW '17}
}

@article{FinocchiMapreduceArboricityJExpAlg,
author = {Finocchi, Irene and Finocchi, Marco and Fusco, Emanuele G.},
title = {Clique Counting in MapReduce: Algorithms and Experiments},
year = {2015},
issue_date = {2015},
publisher = {Association for Computing Machinery},
address = {New York, NY, USA},
volume = {20},
issn = {1084-6654},
url = {https://doi.org/10.1145/2794080},
doi = {10.1145/2794080},
journal = {ACM J. Exp. Algorithmics},
month = oct,
articleno = {1.7},
numpages = {20}
}

@inproceedings{SuriVasCursedLastReducerWWW11,
author = {Suri, Siddharth and Vassilvitskii, Sergei},
title = {Counting triangles and the curse of the last reducer},
year = {2011},
isbn = {9781450306324},
publisher = {Association for Computing Machinery},
address = {New York, NY, USA},
url = {https://doi.org/10.1145/1963405.1963491},
doi = {10.1145/1963405.1963491},
booktitle = {Proceedings of the 20th International Conference on World Wide Web},
pages = {607–614},
numpages = {8},
location = {Hyderabad, India},
series = {WWW '11}
}

@article{AtseriasDatabaseJoinJComp,
author = {Atserias, Albert and Grohe, Martin and Marx, D\'{a}niel},
title = {Size Bounds and Query Plans for Relational Joins},
journal = {SIAM Journal on Computing},
volume = {42},
number = {4},
pages = {1737-1767},
year = {2013},
doi = {10.1137/110859440},
URL = {https://doi.org/10.1137/110859440},
eprint = {https://doi.org/10.1137/110859440}
}

@article{watts_collective_1998,
  author = {Watts, Duncan J. and Strogatz, Steven H.},
  doi = {10.1038/30918},
  journal = {Nature},
  number = 6684,
  pages = {440--442},
  title = {Collective dynamics of ‘small-world’ networks},
  volume = 393,
  year = 1998
}

@article{LucePerry1949AMO,
  title={A method of matrix analysis of group structure},
  author={Duncan Luce and Albert D. Perry},
  journal={Psychometrika},
  year={1949},
  volume={14},
  pages={95-116},
  url={https://api.semanticscholar.org/CorpusID:16186758}
}

@article{CAggarwalEvolutionaryNetworkSurvey,
author = {Aggarwal, Charu and Subbian, Karthik},
title = {Evolutionary Network Analysis: A Survey},
year = {2014},
issue_date = {July 2014},
publisher = {Association for Computing Machinery},
address = {New York, NY, USA},
volume = {47},
number = {1},
issn = {0360-0300},
url = {https://doi.org/10.1145/2601412},
doi = {10.1145/2601412},
journal = {ACM Comput. Surv.},
month = may,
articleno = {10},
numpages = {36}
}

@article{Eckmann2001CurvatureOC,
  title={Curvature of co-links uncovers hidden thematic layers in the World Wide Web},
  author={Jean-Pierre Eckmann and Elisha Moses},
  journal={Proceedings of the National Academy of Sciences of the United States of America},
  year={2001},
  volume={99},
  pages={5825 - 5829},
  url={https://api.semanticscholar.org/CorpusID:2618949}
}

@inproceedings{LeskovecBKT/SIGKDD08/EvolSocNet,
author = {Leskovec, Jure and Backstrom, Lars and Kumar, Ravi and Tomkins, Andrew},
title = {Microscopic evolution of social networks},
year = {2008},
isbn = {9781605581934},
publisher = {Association for Computing Machinery},
address = {New York, NY, USA},
url = {https://doi.org/10.1145/1401890.1401948},
doi = {10.1145/1401890.1401948},
booktitle = {Proceedings of the 14th ACM SIGKDD International Conference on Knowledge Discovery and Data Mining},
pages = {462–470},
numpages = {9},
location = {Las Vegas, Nevada, USA},
series = {KDD '08}
}

@article{ShinRealWorldKCore/JKInfSyst,
author = {Shin, Kijung and Eliassi-Rad, Tina and Faloutsos, Christos},
title = {Patterns and anomalies in k-cores of real-world graphs with applications},
year = {2018},
issue_date = {March     2018},
publisher = {Springer-Verlag},
address = {Berlin, Heidelberg},
volume = {54},
number = {3},
issn = {0219-1377},
url = {https://doi.org/10.1007/s10115-017-1077-6},
doi = {10.1007/s10115-017-1077-6},
journal = {Knowl. Inf. Syst.},
month = mar,
pages = {677–710},
numpages = {34}
}

@inproceedings{DanischBS/WWW18.RealWorldDegeneracy,
author = {Danisch, Maximilien and Balalau, Oana and Sozio, Mauro},
title = {Listing k-cliques in Sparse Real-World Graphs*},
year = {2018},
isbn = {9781450356398},
publisher = {International World Wide Web Conferences Steering Committee},
address = {Republic and Canton of Geneva, CHE},
url = {https://doi.org/10.1145/3178876.3186125},
doi = {10.1145/3178876.3186125},
booktitle = {Proceedings of the 2018 World Wide Web Conference},
pages = {589–598},
numpages = {10},
location = {Lyon, France},
series = {WWW '18}
}

@InProceedings{beraCG/BoundedArboricity:LIPIcs.ICALP.2020.11,
  author =	{Bera, Suman K. and Chakrabarti, Amit and Ghosh, Prantar},
  title =	{{Graph Coloring via Degeneracy in Streaming and Other Space-Conscious Models}},
  booktitle =	{47th International Colloquium on Automata, Languages, and Programming (ICALP 2020)},
  pages =	{11:1--11:21},
  series =	{Leibniz International Proceedings in Informatics (LIPIcs)},
  ISBN =	{978-3-95977-138-2},
  ISSN =	{1868-8969},
  year =	{2020},
  volume =	{168},
  editor =	{Czumaj, Artur and Dawar, Anuj and Merelli, Emanuela},
  publisher =	{Schloss Dagstuhl -- Leibniz-Zentrum f{\"u}r Informatik},
  address =	{Dagstuhl, Germany},
  URL =		{https://drops.dagstuhl.de/entities/document/10.4230/LIPIcs.ICALP.2020.11},
  URN =		{urn:nbn:de:0030-drops-124182},
  doi =		{10.4230/LIPIcs.ICALP.2020.11}
}

@inproceedings{GoelGustedtBoundedArboricity,
author = {Goel, Gaurav and Gustedt, Jens},
title = {Bounded arboricity to determine the local structure of sparse graphs},
year = {2006},
isbn = {3540483810},
publisher = {Springer-Verlag},
address = {Berlin, Heidelberg},
url = {https://doi.org/10.1007/11917496_15},
doi = {10.1007/11917496_15},
booktitle = {Proceedings of the 32nd International Conference on Graph-Theoretic Concepts in Computer Science},
pages = {159–167},
numpages = {9},
location = {Bergen, Norway},
series = {WG'06}
}

@article{OnakSSW/LowArboricityMinorFree/ACMTrAlg,
author = {Onak, Krzysztof and Schieber, Baruch and Solomon, Shay and Wein, Nicole},
title = {Fully Dynamic MIS in Uniformly Sparse Graphs},
year = {2020},
issue_date = {April 2020},
publisher = {Association for Computing Machinery},
address = {New York, NY, USA},
volume = {16},
number = {2},
issn = {1549-6325},
url = {https://doi.org/10.1145/3378025},
doi = {10.1145/3378025},
journal = {ACM Trans. Algorithms},
month = mar,
articleno = {26},
numpages = {19}
}

@InProceedings{LowArboricityMatching/FSTTCS2024/Konrad,
  author =	{Konrad, Christian and McGregor, Andrew and Sengupta, Rik and Than, Cuong},
  title =	{{Matchings in Low-Arboricity Graphs in the Dynamic Graph Stream Model}},
  booktitle =	{44th IARCS Annual Conference on Foundations of Software Technology and Theoretical Computer Science (FSTTCS 2024)},
  pages =	{29:1--29:15},
  series =	{Leibniz International Proceedings in Informatics (LIPIcs)},
  ISBN =	{978-3-95977-355-3},
  ISSN =	{1868-8969},
  year =	{2024},
  volume =	{323},
  editor =	{Barman, Siddharth and Lasota, S{\l}awomir},
  publisher =	{Schloss Dagstuhl -- Leibniz-Zentrum f{\"u}r Informatik},
  address =	{Dagstuhl, Germany},
  URL =		{https://drops.dagstuhl.de/entities/document/10.4230/LIPIcs.FSTTCS.2024.29},
  URN =		{urn:nbn:de:0030-drops-222187},
  doi =		{10.4230/LIPIcs.FSTTCS.2024.29}
}

@inproceedings{DoryGI/BoundedArboricity/PODC22,
author = {Dory, Michal and Ghaffari, Mohsen and Ilchi, Saeed},
title = {Near-Optimal Distributed Dominating Set in Bounded Arboricity Graphs},
year = {2022},
isbn = {9781450392624},
publisher = {Association for Computing Machinery},
address = {New York, NY, USA},
url = {https://doi.org/10.1145/3519270.3538437},
doi = {10.1145/3519270.3538437},
booktitle = {Proceedings of the 2022 ACM Symposium on Principles of Distributed Computing},
pages = {292–300},
numpages = {9},
location = {Salerno, Italy},
series = {PODC'22}
}

@article{TsourakakisJournalOG,
title={Triangle Sparsifiers},
volume={15},
url={https://jgaa.info/index.php/jgaa/article/view/paper245},
DOI={10.7155/jgaa.00245},
number={6},
journal={Journal of Graph Algorithms and Applications}, 
author={Tsourakakis, Charalampos and Kolountzakis, Mihail and Miller,
        Gary},
year={2011},
month={Oct.},
pages={703–726}}

@article{Hasan2018TriangleCI,
author = {Al Hasan, Mohammad and Dave, Vachik S.},
title = {Triangle counting in large networks: a review},
journal = {WIREs Data Mining and Knowledge Discovery},
volume = {8},
number = {2},
pages = {e1226},
doi = {https://doi.org/10.1002/widm.1226},
year = {2018}
}

@article{chakrabarti2015data,
  title={Data stream algorithms},
  author={Chakrabarti, Amit},
  journal={Computer Science},
  volume={49},
  pages={149},
  year={2015}
}

@inproceedings{LiuSeshadri/PODC/2024/MPCArboricityTriangle,
author = {Liu, Quanquan C. and Seshadhri, C.},
title = {Brief Announcement: Improved Massively Parallel Triangle Counting in O(1) Rounds},
year = {2024},
isbn = {9798400706684},
publisher = {Association for Computing Machinery},
address = {New York, NY, USA},
url = {https://doi.org/10.1145/3662158.3662819},
doi = {10.1145/3662158.3662819},
booktitle = {Proceedings of the 43rd ACM Symposium on Principles of Distributed Computing},
pages = {519–522},
numpages = {4},
location = {Nantes, France},
series = {PODC '24}
}

@article{MatulaBeck/JACM/83/SmallestLastOrderingArboricity,
author = {Matula, David W. and Beck, Leland L.},
title = {Smallest-last ordering and clustering and graph coloring algorithms},
year = {1983},
issue_date = {July 1983},
publisher = {Association for Computing Machinery},
address = {New York, NY, USA},
volume = {30},
number = {3},
issn = {0004-5411},
url = {https://doi.org/10.1145/2402.322385},
doi = {10.1145/2402.322385},
journal = {J. ACM},
month = jul,
pages = {417–427},
numpages = {11}
}

@article{CHROBAK1991243,
title = {Planar orientations with low out-degree and compaction of adjacency matrices},
journal = {Theoretical Computer Science},
volume = {86},
number = {2},
pages = {243-266},
year = {1991},
issn = {0304-3975},
doi = {https://doi.org/10.1016/0304-3975(91)90020-3},
url = {https://www.sciencedirect.com/science/article/pii/0304397591900203},
author = {Marek Chrobak and David Eppstein}
}

@article{BerettaTetek/TALG/2024/BetterSumEstimationWeightedSampling,
author = {Beretta, Lorenzo and T\v{e}tek, Jakub},
title = {Better Sum Estimation via Weighted Sampling},
year = {2024},
issue_date = {July 2024},
publisher = {Association for Computing Machinery},
address = {New York, NY, USA},
volume = {20},
number = {3},
issn = {1549-6325},
url = {https://doi.org/10.1145/3650030},
doi = {10.1145/3650030},
journal = {ACM Trans. Algorithms},
month = jun,
articleno = {27},
numpages = {33}
}

@inproceedings{BeraC/stacs/2017/CliqueCountingEA,
  author       = {Suman K. Bera and
                  Amit Chakrabarti},
  editor       = {Heribert Vollmer and
                  Brigitte Vall{\'{e}}e},
  title        = {Towards Tighter Space Bounds for Counting Triangles and Other Substructures
                  in Graph Streams},
  booktitle    = {34th Symposium on Theoretical Aspects of Computer Science, {STACS}
                  2017, March 8-11, 2017, Hannover, Germany},
  series       = {LIPIcs},
  volume       = {66},
  pages        = {11:1--11:14},
  publisher    = {Schloss Dagstuhl - Leibniz-Zentrum f{\"{u}}r Informatik},
  year         = {2017},
  url          = {https://doi.org/10.4230/LIPIcs.STACS.2017.11},
  doi          = {10.4230/LIPICS.STACS.2017.11},
  bibsource    = {dblp computer science bibliography, https://dblp.org}
}

@inproceedings{BeraSesh/pods/2020/DegeneracyStream,
  author       = {Suman K. Bera and
                  C. Seshadhri},
  editor       = {Dan Suciu and
                  Yufei Tao and
                  Zhewei Wei},
  title        = {How the Degeneracy Helps for Triangle Counting in Graph Streams},
  booktitle    = {Proceedings of the 39th {ACM} {SIGMOD-SIGACT-SIGAI} Symposium on Principles
                  of Database Systems, {PODS} 2020, Portland, OR, USA, June 14-19, 2020},
  pages        = {457--467},
  publisher    = {{ACM}},
  year         = {2020},
  url          = {https://doi.org/10.1145/3375395.3387665},
  doi          = {10.1145/3375395.3387665},
  bibsource    = {dblp computer science bibliography, https://dblp.org}
}

@inproceedings{KallaugherP/soda/2017/HybridTriangleStream,
  author       = {John Kallaugher and
                  Eric Price},
  editor       = {Philip N. Klein},
  title        = {A Hybrid Sampling Scheme for Triangle Counting},
  booktitle    = {Proceedings of the Twenty-Eighth Annual {ACM-SIAM} Symposium on Discrete
                  Algorithms, {SODA} 2017, Barcelona, Spain, Hotel Porta Fira, January
                  16-19},
  pages        = {1778--1797},
  publisher    = {{SIAM}},
  year         = {2017},
  url          = {https://doi.org/10.1137/1.9781611974782.116},
  doi          = {10.1137/1.9781611974782.116},
  bibsource    = {dblp computer science bibliography, https://dblp.org}
}

@inproceedings{KallaugherKP/focs/2018/SketchingSubgraphCounting,
  author       = {John Kallaugher and
                  Michael Kapralov and
                  Eric Price},
  editor       = {Mikkel Thorup},
  title        = {The Sketching Complexity of Graph and Hypergraph Counting},
  booktitle    = {59th {IEEE} Annual Symposium on Foundations of Computer Science, {FOCS}
                  2018, Paris, France, October 7-9, 2018},
  pages        = {556--567},
  publisher    = {{IEEE} Computer Society},
  year         = {2018},
  url          = {https://doi.org/10.1109/FOCS.2018.00059},
  doi          = {10.1109/FOCS.2018.00059},
  bibsource    = {dblp computer science bibliography, https://dblp.org}
}

@inproceedings{KallaugherP/stoc/2020/SketchingVsTurnstile,
  author       = {John Kallaugher and
                  Eric Price},
  editor       = {Konstantin Makarychev and
                  Yury Makarychev and
                  Madhur Tulsiani and
                  Gautam Kamath and
                  Julia Chuzhoy},
  title        = {Separations and equivalences between turnstile streaming and linear
                  sketching},
  booktitle    = {Proceedings of the 52nd Annual {ACM} {SIGACT} Symposium on Theory
                  of Computing, {STOC} 2020, Chicago, IL, USA, June 22-26, 2020},
  pages        = {1223--1236},
  publisher    = {{ACM}},
  year         = {2020},
  url          = {https://doi.org/10.1145/3357713.3384278},
  doi          = {10.1145/3357713.3384278},
  bibsource    = {dblp computer science bibliography, https://dblp.org}
}

@inproceedings{KallaugherMPV/pods/2019/ALModelCycleCount,
  author       = {John Kallaugher and
                  Andrew McGregor and
                  Eric Price and
                  Sofya Vorotnikova},
  editor       = {Dan Suciu and
                  Sebastian Skritek and
                  Christoph Koch},
  title        = {The Complexity of Counting Cycles in the Adjacency List Streaming
                  Model},
  booktitle    = {Proceedings of the 38th {ACM} {SIGMOD-SIGACT-SIGAI} Symposium on Principles
                  of Database Systems, {PODS} 2019, Amsterdam, The Netherlands, June
                  30 - July 5, 2019},
  pages        = {119--133},
  publisher    = {{ACM}},
  year         = {2019},
  url          = {https://doi.org/10.1145/3294052.3319706},
  doi          = {10.1145/3294052.3319706},
  bibsource    = {dblp computer science bibliography, https://dblp.org}
}

@inproceedings{EdenRS/STOC/2018/SublinearCliques,
author = {Eden, Talya and Ron, Dana and Seshadhri, C.},
title = {On approximating the number of k-cliques in sublinear time},
year = {2018},
isbn = {9781450355599},
publisher = {Association for Computing Machinery},
address = {New York, NY, USA},
url = {https://doi.org/10.1145/3188745.3188810},
doi = {10.1145/3188745.3188810},
booktitle = {Proceedings of the 50th Annual ACM SIGACT Symposium on Theory of Computing},
pages = {722–734},
numpages = {13},
location = {Los Angeles, CA, USA},
series = {STOC 2018}
}

\end{document}